\documentclass[letterpaper, UKenglish, cleveref, autoref, thm-restate, numberwithinsect, nolineno]{socg-lipics-v2021}

\hideLIPIcs  

\title{Dudeney's Dissection is Optimal}

\author{Erik {D. Demaine}}
    {Massachusetts Institute of Technology, USA \and \url{https://erikdemaine.org/}
    }{edemaine@mit.edu}{0000-0003-3803-5703}{}
\author{Tonan {Kamata}}
    {Japan Advanced Institute of Science and Technology, Japan \and 
    \url{https://researchmap.jp/TonanKamata} }{kamata@jaist.ac.jp}{0009-0001-9786-9751}
    {is mainly supported by JSPS KAKENHI Grants Number 22J10261 and 24K23857, and partially supported by Grants Number 20H05961 and 20H05964.}

\author{Ryuhei {Uehara}}{Japan Advanced Institute of Science and Technology, Japan \and             \url{https://www.jaist.ac.jp/~uehara/}
    }{uehara@jaist.ac.jp}{0000-0003-0895-3765}
    {is mainly supported by JSPS KAKENHI Grants Number 24H00690 and partially supported by Grants Number 22H01423, 20H05961, and 20H05964.}

\authorrunning{E. D. Demaine, T. Kamata and R. Uehara} 

\Copyright{E. D. Demaine, T. Kamata and R. Uehara} 


\category{} 

\relatedversion{} 

\EventEditors{0}
\EventNoEds{0}
\EventLongTitle{41st International Symposium on Computational Geometry (SoCG 2025)}
\EventShortTitle{SoCG 2025}
\EventAcronym{SoCG}
\EventYear{2025}
\EventDate{June 23--27, 2025}
\EventLocation{kanazawa, Japan}
\EventLogo{}
\SeriesVolume{}
\ArticleNo{}

\newcommand{\ED}{\mathscr{ED}}
\newcommand{\VD}{\mathscr{VD}}

\let\epsilon=\varepsilon

\def\defn#1{\textbf{\textit{\boldmath #1}}}

\begin{document}

\maketitle
\begin{abstract}
  In 1907, Henry Ernest Dudeney posed a puzzle: ``cut any equilateral triangle \dots\ into as few pieces as possible that will fit together and form a perfect square'' (without overlap, via translation and rotation).
  Four weeks later, Dudeney demonstrated a beautiful four-piece solution,
  which today remains perhaps the most famous example of dissection.
  In this paper (over a century later), we finally solve Dudeney's puzzle,
  by proving that the equilateral triangle and square have no common dissection with three or fewer polygonal pieces.
  We reduce the problem to the analysis of discrete graph structures representing the correspondence between the edges and the vertices of the pieces forming each polygon.
\end{abstract}

\section{Introduction}
    \defn{Dissection} \cite{Lindgren-1972,Frederickson-1997}
    is the process of transforming one shape $A$ into another shape $B$
    by cutting $A$ into pieces and re-arranging those pieces to form~$B$
    (without overlap).  Necessarily, $A$ and $B$ must have the same area
    (a property Euclid used to prove area equalities).
    Conversely, every two polygons of the same area have a dissection ---
    a result from over two centuries ago
    \cite{Lowry-1814,Wallace-1831,Bolyai-1832,Gerwien-1833}.
    The number of pieces cannot be bounded as a function of the number
    of vertices (consider, for example, dissecting
    an $N \times N$ square into an $N^2 \times 1$ rectangle),
    but there is a pseudopolynomial upper bound \cite{HingedDissection}.

    Given two specific polygons, can we find the dissection between them
    with the \emph{fewest possible pieces}?
    In general, this minimization problem is NP-hard, even to approximate within
    a factor of $1 + 1/1080 - \epsilon$ \cite{DissectionHard}.
    Worse, the existence of a $k$-piece dissection is not known to be decidable,
    even for $k=2$ pieces \cite{DissectionHard}.
    A core difficulty is that the number of cut segments has no obvious upper bound,
    even for 2-piece dissection.
    The special case where the $k=2$ pieces must be mirror congruent to each
    other has a polynomial-time decision algorithm
    \cite{El-Khechen-Iacono-Fevens-2008},
    even though the number of cuts cannot be bounded as a function of~$n$
    \cite{Rote-1997}.


    Instances of the dissection problem have captivated puzzle creators
    and solvers for centuries;
    see \cite{Frederickson-1997,Gardner-1969-dissection} for much history.
    One pioneer was English puzzlist Henry Ernest Dudeney
    \cite{Gardner-1961-dudeney}, who
    published many dissection puzzles in newspapers and magazines
    in the late 19th and early 20th centuries
    \cite{Dudeney-1908-canterbury,Dudeney-1917-amusements}.
    Over the years, geometric puzzlers continually improved the records
    for the fewest-piece dissections between various pairs of polygons.
    Harry Lindgren (an Australian patent reviewer) collected many of these
    results and improved most of them in his book \cite{Lindgren-1972}.
    He wrote:
    \begin{quote}
      ``In a few cases (very, very few)
      it could perhaps be rigorously proved that the minimum number
      has been attained, and in a few more one can feel morally
      certain; in all the rest it is possible that you may find a dissection
      that is better than those already known''
      \cite[p.~1]{Lindgren-1972}
    \end{quote}

    Some past work determines the asymptotic growth of the number of pieces
    for certain infinite families.  Cohn \cite{Cohn-1975} showed that a
    triangle of diameter $d$ needs $\Theta(d)$ pieces to dissect into a unit square,
    where the constant in the $\Theta$ is between $0.7415$ and $1$.
    (The obvious diameter lower bound is $d/\sqrt 2 \approx 0.7071 \, d$.)
    Kranakis, Krizanc, and Urrutia \cite{Kranakis-Krizanc-Urrutia-2000}
    showed that the regular $n$-gon needs $\Theta(n)$ pieces to dissect into a
    square, where the constant in the $\Theta$ is between $1/4$ and $1/2$.

    In this paper, we give the first nontrivial proof of exact optimality of
    the number of pieces in a dissection,
    by proving optimality of the famous four-piece dissection
    of an equilateral triangle to a square shown in Figure~\ref{fig:dudeney}.
    Dudeney \cite{dudeney1902} posed this dissection as a puzzle on April 6, 1902,
    without making it especially clear whether he had a solution.
    In the next issue of his column (April 20), he described an easy five-piece
    solution, wrote that Mr.\ C. W. McElroy of Manchester had found a four-piece
    solution, and gave readers another two weeks to try to find it.
    No one did, leading Dudeney to conclude that ``the puzzle may be regarded as
    a decidedly hard nut'' in the next issue where he gave the four-piece
    solution (May 4).  It remains unclear whether this solution was
    originally invented by Dudeney or McElroy
    \cite{Frederickson-1997,Frederickson-2002}.
    The puzzle and solution appeared later in Dudeney's book as
    ``The Haberdasher's Puzzle'' \cite[Puzzle 26]{Dudeney-1908-canterbury};
    Gardner \cite{Gardner-1961-dudeney}
    called it ``Dudeney's best-known geometrical discovery''.

    \begin{figure}[t]\centering
        \includegraphics[width=0.5\textwidth]{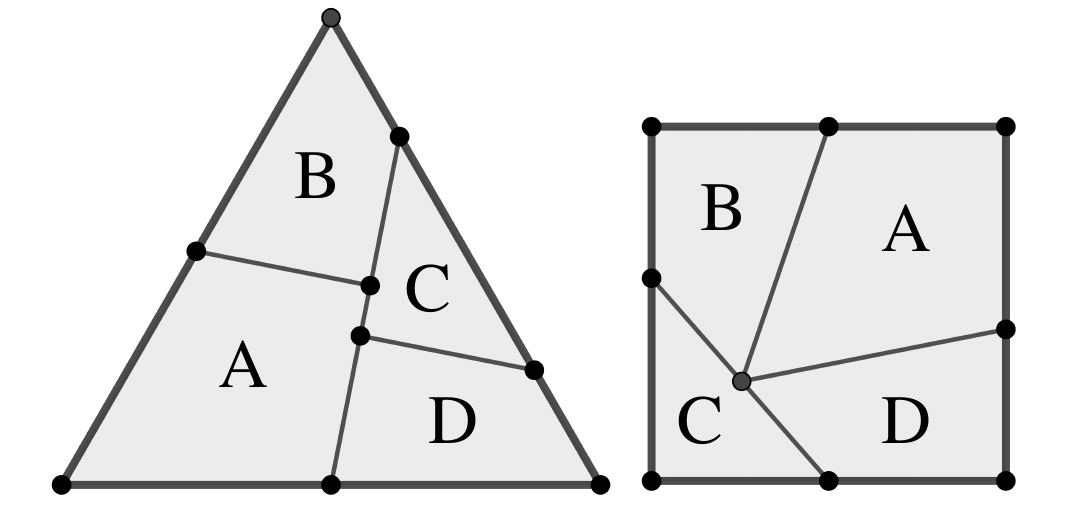}
        \caption{The four-piece dissection between an equilateral triangle and a square \cite{dudeney1902}.}
        \label{fig:dudeney}
    \end{figure}
        
    Since this discovery over 120 years ago, it has remained open whether this
    dissection is optimal, or whether the record will someday be improved upon
    (as many other records have)
    \cite{dudeney1902,HingedDissection,DissectionHard}.
    In this paper, we finally settle this problem:

    \begin{theorem} \label{thm:main}
        There is no dissection with three or fewer polygonal pieces between a square and an equilateral triangle, when we forbid flipping pieces.
    \end{theorem}

    Figure~\ref{fig:dudeney} also does not require the pieces to be
    flipped, but this is not necessarily a requirement in the original puzzle.
    When first teasing ``the correct solution'', Dudeney wrote,
    ``surprising as it may seem, it is not necessary to turn over any piece.''
    So we leave an intriguing puzzle for future work:
    is there a 3-piece dissection if we allow flipping pieces?
    We also do not know whether curved pieces might help, though we suspect
    Theorem~\ref{thm:main} extends to both of these settings.

\subsection{Overview}

    The structure of the proof is as follows.
    First, it is relatively easy to show that a two-piece dissection is impossible (see Lemma~\ref{lem:geom-T-vertex}).
    Thus, we focus on proving the impossibility of a three-piece dissection.
    To do so, we classify the finitely many possible topologies
    for cutting each polygon
    that could potentially result in a three-piece dissection.
    Next, we use fundamental properties of dissections to narrow down
    the feasible combinations of these cutting topologies for the two shapes.
    Finally, we show that all remaining combinations are infeasible
    using the new concepts of ``matching diagrams'',
    which handle the potentially unbounded number of cut segments.

    This paper is organized as follows.
    We start in Section~\ref{sec:preliminaries} with basic definitions of dissections, specifically the graph of cuts within each shape, and a categorization of vertices in these graphs into six types.
    Then, in Section~\ref{sec:matching-diagram}, we introduce our primary tool for analyzing dissections --- \defn{vertex and edge matching diagrams} --- and prove several general necessary properties for dissections to work, which can be applied to other dissection problems as well.
    Next, in Section~\ref{sec:overview-ex}, we give a technical overview of our proof, and illustrate by describing one key case in detail; this section gives a flavor of our techniques without getting bogged down in the many cases necessary to complete the proof (but from a technical perspective, it can be skipped).
    Finally, in Section~\ref{sec:main}, we give the full proof, which applies the matching diagrams from Section~\ref{sec:matching-diagram} to the specific case of the equilateral triangle and the square.

\section{Preliminaries}
\label{sec:preliminaries}
    We define a \defn{polygon} to be a connected region in the plane of finite area bounded by finitely many line segments.
    (This definition does not require polygons to be simple,
    although we will show in Lemma~\ref{lem:geom-simple} that our main proof
    can consider just simple polygons.)
    We refer to a polygon's line segments as \defn{sides} and the intersections of these sides as \defn{corners}.%
    \footnote{We use these terms (instead of more common terms such as
      ``edges'' and ``vertices'') to provide unambiguous terminology
      for edge- and node-like features of various parallel structures:
      polygons, cut graphs, and matching graphs.}
    The \defn{boundary} of a polygon consists of the corners and
    points along the sides,
    while the \defn{interior} is all other points within the region.
    Let $|s|$ denote the Euclidean length of a side~$s$,
    and let $\angle(x)$ denote the \defn{interior angle}
    of the region at corner~$x$.

    A \defn{$k$-piece dissection} of a pair of polygons $(P, P^{\prime})$
    consists of $k$ polygons $P_1, P_2, \ldots, P_k$ and two sets of corresponding congruence transformations ${\lambda_1, \lambda_2, \ldots, \lambda_k}$ and ${\lambda^{\prime}_1, \lambda^{\prime}_2, \ldots, \lambda^{\prime}_k}$ such that $\bigsqcup_{i} \lambda_i(P_i) = P$ and $\bigsqcup_{i} \lambda^{\prime}_i(P_i) = P^{\prime}$.%
    \footnote{The notation $\sqcup$ denotes interior-disjoint set union,
      i.e., we allow the polygons to overlap on shared boundaries.}
    Figure~\ref{fig:dudeney} is an example of a 4-piece dissection.
    We refer to $P$ and $P^{\prime}$ as \defn{target shapes}, and $P_1, P_2, \ldots, P_k$ as \defn{pieces}. 
    Naturally, the areas of the target shapes must be equal to the total area of the pieces.

    A dissection of $(P, P^{\prime})$ divides each target shape $X \in \{P, P^{\prime}\}$ into pieces by cut lines. 
    We define a finite geometric graph called the \defn{cut graph} $G^X$, which represents the cut lines and the boundary of~$X$.
    Specifically, $G^X$ has \defn{faces} that correspond one-to-one with the pieces, \defn{vertices} $V(G^X)$ which consist of the corners of $X$ and the points along the cut lines having at least one surrounding angle not equal to $\pi$, and \defn{edges} $E(G^X)$ which consist of the cut and boundary lines connecting these vertices.
    For example, Figure~\ref{fig:dudeney} highlights the vertices and edges of $G^P$ and $G^{P^{\prime}}$ for Dudeney’s four-piece dissection.
    The edges $E(G^X)$ either originate from cut lines, referred to as \defn{internal edges}, or from the sides of~$X$, referred to as \defn{boundary edges}.
    The vertices $V(G^X)$ are intersections between at least two edges;
    a vertex met only by internal edges is an \defn{internal vertex}, while
    a vertex met by at least one boundary edge is a \defn{boundary vertex}.
    Based on the elements surrounding each vertex of $G^X$, we classify boundary vertices into Types 1--2 and internal vertices into Types 3--6 as follows (refer to Figure~\ref{fig:types}):
    \begin{itemize}
        \item \textbf{Type 1: Corner vertex.} A boundary vertex corresponding to a corner of the target shape $X$.
        \item \textbf{Type 2: Side vertex.} A boundary vertex corresponding to a point along a side of the target shape~$X$.
        \item \textbf{Type 3: Paired vertex.} An internal vertex where exactly two piece corners meet.
        \item \textbf{Type 4: Convex vertex.} An internal vertex where three or more piece corners meet, and all interior angles are convex.
        \item \textbf{Type 5: Reflex vertex.} An internal vertex where three or more piece corners meet and only one of them is reflex.
        \item \textbf{Type 6: Flat vertex.} An internal vertex where at least two piece corners and one piece side meet.
    \end{itemize}
    \begin{figure}[htbp]
        \centering
        \includegraphics[width=1\textwidth]{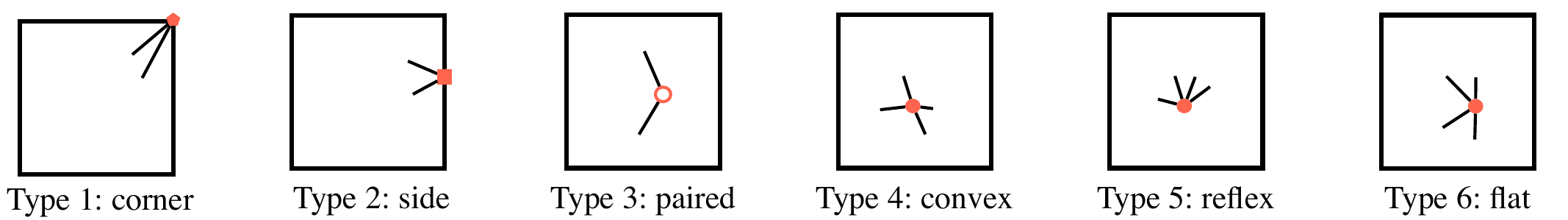}
        \caption{The types of vertices in $G^X$.}
        \label{fig:types}
    \end{figure}
    This distinction will be useful in Section~\ref{sec:enum} when counting the number of angles appearing in the cut graph.

    We define a \defn{subdivision} of the cut graph $G^X$ to be a graph obtained by replacing internal edges of $G^X$ with paths consisting of paired (Type 3) vertices.  
    We define an equivalence relation $\sim$ where $G^X_0 \sim G^X_1$ if
    there exist subdivisions $H^X_0$ and $H^X_1$ of the graphs $G^X_0$ and $G^X_1$ respectively, and a bijection $\phi: V(H^X_0) \rightarrow V(H^X_1)$ such that $\phi$ is a graph isomorphism that disregards geometry while preserving vertex types.

\section{Matching Diagram}\label{sec:matching-diagram}
    When a set of pieces is assembled to form a target shape, each side of a piece corresponds to an edge of the cut graph, and each corner corresponds to a vertex of the cut graph. 
    Therefore, when a pair of target shapes has a dissection, there are two corresponding relationships: one between the edges and another between the vertices of the cut graph.
    In this section, we formalize these relationships and describe these properties.  
\subsection{Edge Matching Diagram}\label{sec:EMD}
    When analyzing the relationship between the sides of pieces and the edges of the cut graph, the fundamental principle is that boundary edges correspond to a single side of a piece, while interior edges correspond to two sides from different pieces. 
    However, this correspondence becomes more complex in the vicinity of flat vertices.
    For instance, when there is a flat vertex on the side of a piece, two sides of the pieces correspond to this side (see the left of Figure~\ref{fig:around-flat}).
    Moreover, when multiple flat vertices are adjacent on the graph, more complex correspondences arise (see the middle and right of Figure~\ref{fig:around-flat}).

    \begin{figure}[htbp]\centering
        \includegraphics[width=1\textwidth]{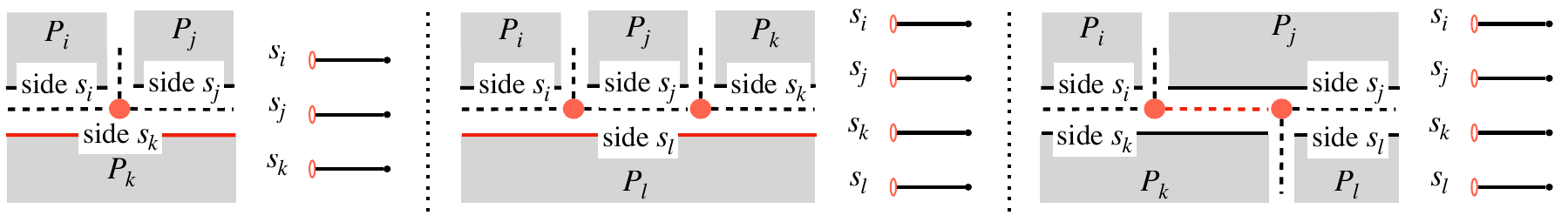}
        \caption{Examples of flat vertices affecting the edge correspondence, and the corresponding nodes in the edge matching diagram.}
        \label{fig:around-flat}
    \end{figure}
    
    To systematically represent these relationships, we introduce a bipartite diagram. 
    In this diagram, each side of a piece is treated as a link, and the edges of the cut graph are treated as nodes. 
    The endpoints of each link are the nodes representing the edges, which correspond to the sides indicated by the links.
    This construction ensures that an edge corresponding one-to-one to two sides becomes a degree-2 node, while all other edges become degree-1 nodes. 
    In cases where a side crosses multiple edges due to flat vertices, we introduce virtual nodes (see the red lines in the left and the middle of Figure~\ref{fig:around-flat}). 
    In contrast, edges of the cut graph that do not correspond to any piece side are removed from the node set (see the red dotted line in the right of Figure~\ref{fig:around-flat}).
    We now formalize this construction as follows:
    \begin{definition}
        We define an \defn{edge matching diagram} $\ED$ for a pair of cut graphs $G^P$ and $G^{P^{\prime}}$ (see Figure~\ref{fig:EMG}) as follows:
        \begin{itemize}
            \item For each $X \in \{P, P^{\prime}\}$, define the node set $N^{X}(\ED)$ as follows:
            \begin{itemize}
                \item Initialize $N^{X}(\ED)$ as $E(G^X)$.
                \item For an edge $x \in N^{X}(\ED)$, if both endpoints of $x$ in the cut graph are flat vertices and an angle of $\pi$ appears on opposite sides of $x$, then remove $x$ from $N^{X}(\ED)$.
                \item For each side of a piece $x$, if there is a flat vertex on $x$  when forming $X$, add a new corresponding node to $N^{X}(\ED)$.
            \end{itemize}
            \item Let $L(\ED)$ be the set of links between $N^{P}(\ED)$ and $N^{P^{\prime}}(\ED)$, connecting a link for each side $e \in E(P_1) \cup E(P_2) \cup \cdots \cup E(P_k)$ according to the correspondence between links representing the sides of pieces and nodes representing the edges of the graph.
        \end{itemize}
        \begin{figure}[htbp]\centering
            \includegraphics[width=0.8\textwidth]{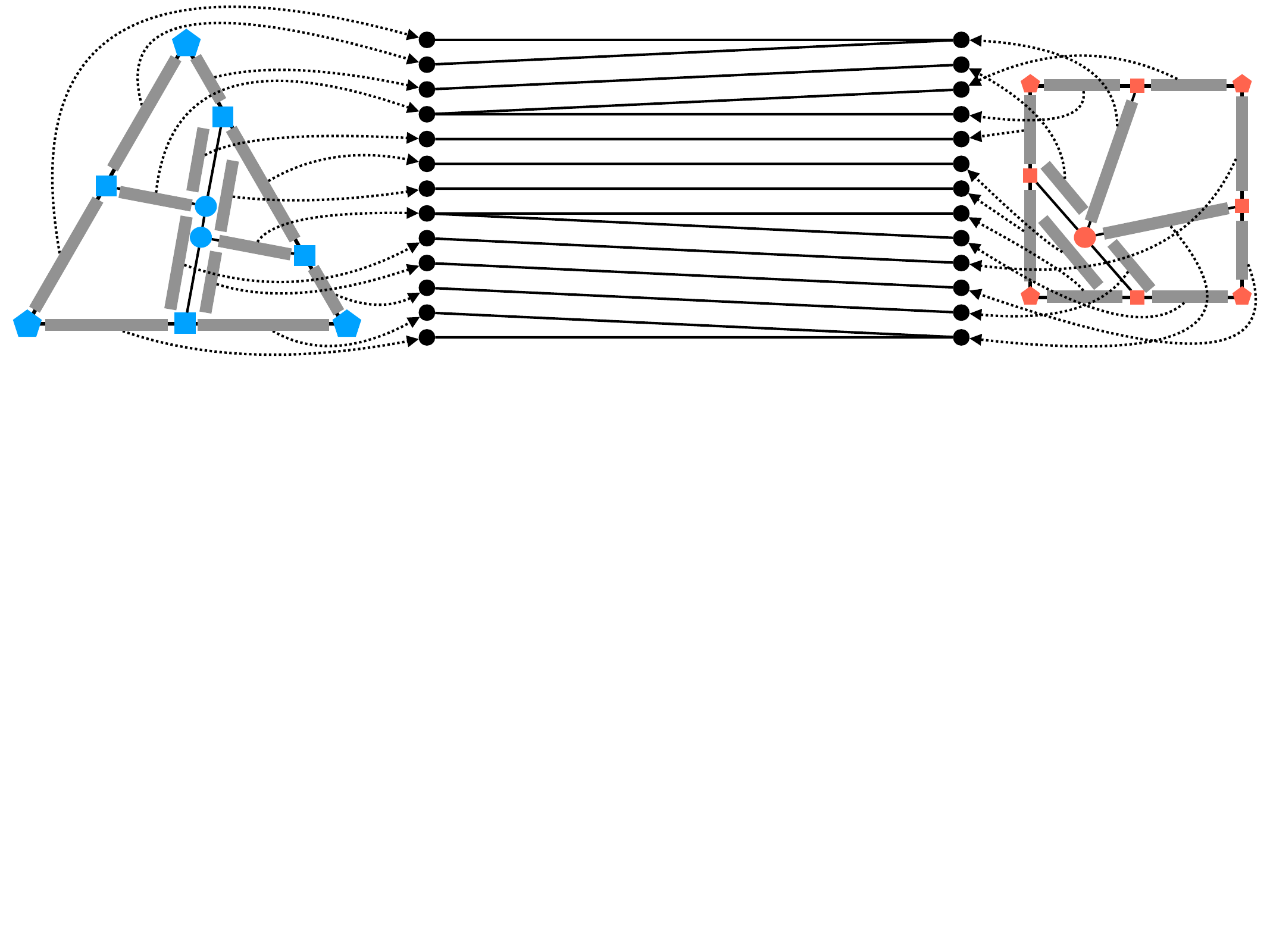}
            \caption{The edge matching diagram of Dudeney's dissection from Figure~\ref{fig:dudeney}.}
            \label{fig:EMG}
        \end{figure}
    \end{definition}
    The added nodes corresponding to a flat vertex are called \defn{flat nodes}, nodes corresponding to boundary edges are called \defn{boundary nodes}, and all other nodes are called \defn{internal nodes}. 
    We use terms such as degree, leaf, and path for the diagram in the same way as they are used in standard graph theory.

    By construction, the degree of each node in $\ED$ is either $1$ or $2$.
    Moreover, since two sides sharing a degree-2 node correspond to the same edge of a cut graph, their lengths are equal.
    Consequently, a path can be obtained by iteratively following degree-2 nodes starting from a leaf of $\ED$; this path satisfies the following property:

    \begin{observation}\label{obs:EG-length}
        For any path in $\ED$, all piece sides corresponding to the nodes along the path have equal lengths.
    \end{observation}
    
\subsection{Vertex Matching Diagram}\label{sec:VMD}
    Next, we consider the relationships between the corners of the pieces and the vertices of the cut graph. 
    This relationship can be considered more simply compared to the relationships between edges discussed in Section~\ref{sec:EMD}.
    A corner of a piece corresponds to a vertex in the cut graph $G^X$ of each target shape $X \in \{P, P^{\prime}\}$. 
    Conversely, when looking at the vertices of $G^X$, one or multiple piece corners may gather.
    To capture this relationship, we introduce the following definition:
    \begin{definition}
        We define a \defn{vertex matching diagram} $\VD$ for $G^P$ and $G^{P^{\prime}}$ (see Figure~\ref{fig:VMG}):
        \begin{itemize}
            \item For each $X \in \{P, P^{\prime}\}$, define the node set $N^{X}(\VD)$ as $V(G^X)$.
            \item Let $L(\VD)$ be the set of links between $N^{P}(\VD)$ and $N^{P^{\prime}}(\VD)$ by connecting a link for each corner $x \in V(P_1) \cup V(P_2) \cup \cdots \cup V(P_k)$ according to the correspondence between links representing the corners of pieces and nodes representing the vertices of the graph.
        \end{itemize}
        \begin{figure}[htbp]\centering
            \includegraphics[width=0.8\textwidth]{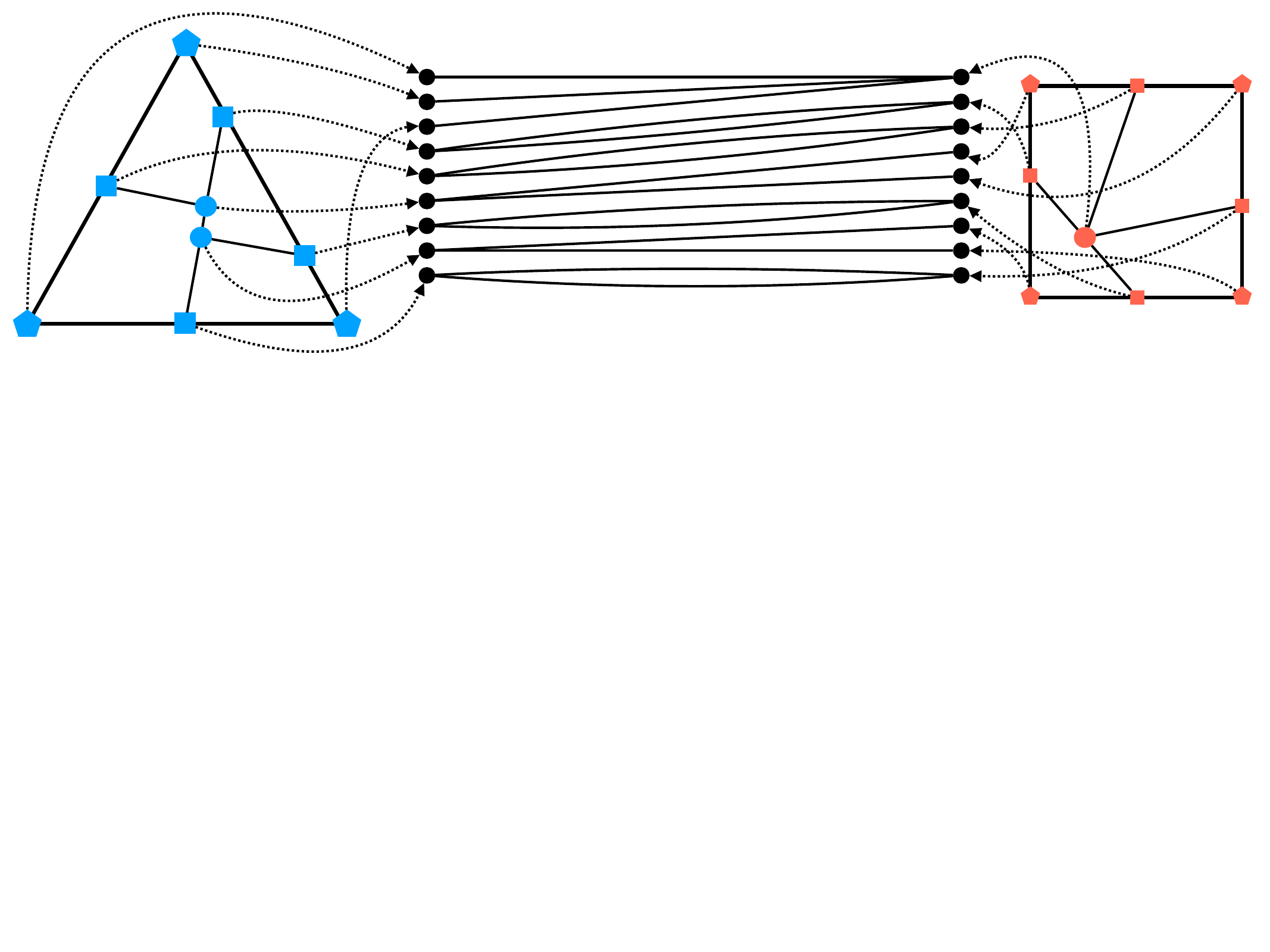}
            \caption{The vertex matching diagram of Dudeney's dissection from Figure~\ref{fig:dudeney}.}
            \label{fig:VMG}
        \end{figure}
    \end{definition}
    Since each node of $\VD$ corresponds to a vertex of the cut graph, we define the terms \defn{corner}, \defn{side}, \defn{paired}, \defn{convex}, \defn{reflex}, and \defn{flat} directly for the nodes as well.

    We define a weight $\omega(x)$ for each node $x$ of $\VD$ by the sum of the interior angles of the corners that gather at $v$. 
    The value of $\omega(x)$ is $2\pi$ if $x$ is a paired, convex, or reflex vertex.
    In the case that $x$ is a flat or side node, $\omega(x)$ is $\pi$.
    In the case that $x$ is a corner node, $\omega(x)$ is the interior angle of the corresponding corner of the target shape.
    
    We can also observe the following by the fact that a paired vertex, which is shared by two corners, has a weight of $2\pi$:
    \begin{observation}\label{obs:VG-angle}
        For any path in $\VD$ that consists of paired vertices, $\angle(x) + \angle(x^{\prime}) = 2\pi$ or $\angle(x) = \angle(x^{\prime})$ holds depending on whether the length of the path is odd or even, where $x$ and $x^{\prime}$ are the corners corresponding to the endpoint nodes of the path.
    \end{observation}
        \begin{figure}[htbp]\centering
            \includegraphics[width=0.75\textwidth]{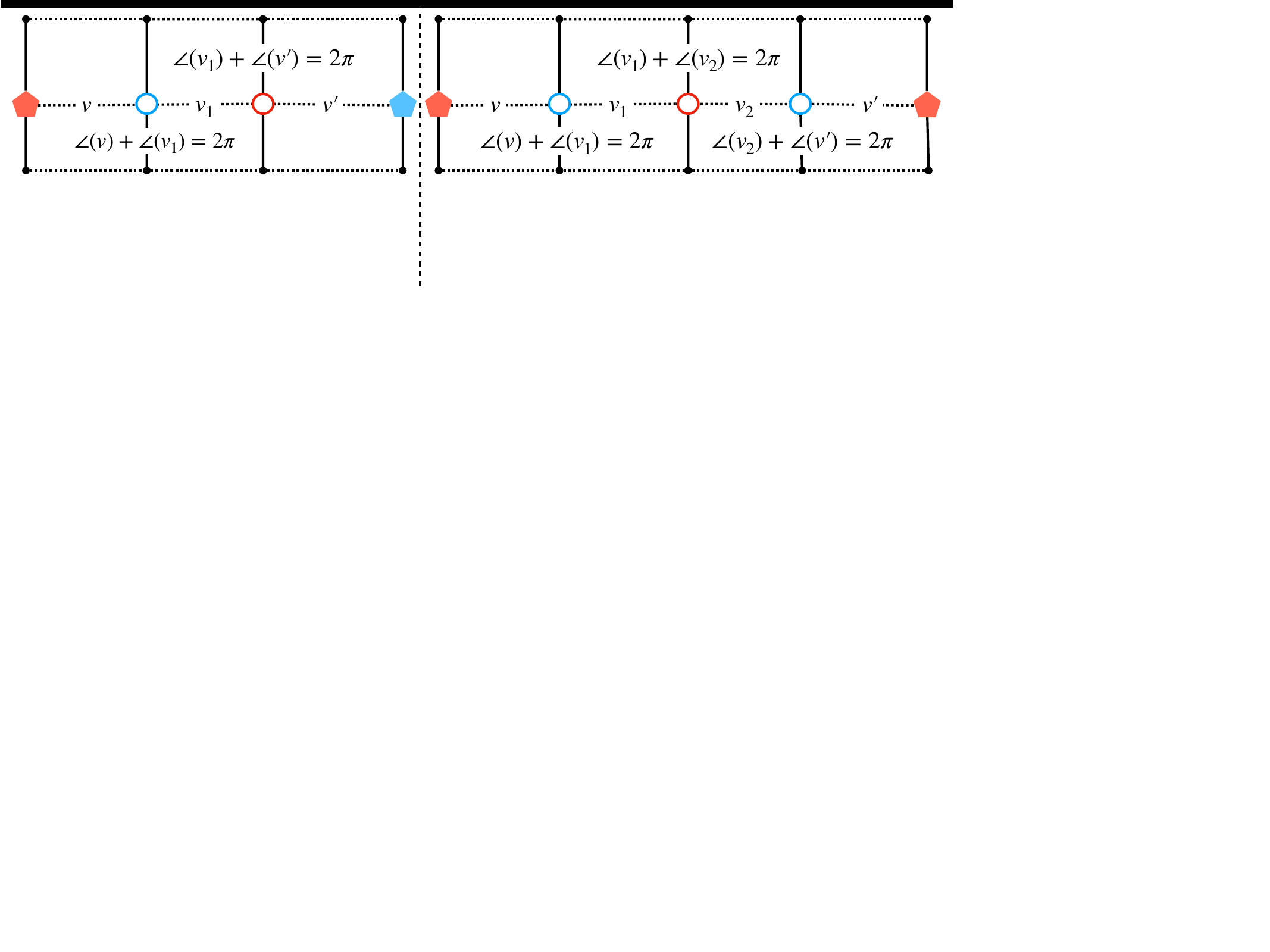}
            \caption{Relationship between the interior angles in a path of $\VD$ composed of paired vertices. The red and blue marks represent the nodes of $\VD$, and the dotted lines connecting them represent links, that is, corners of the pieces.}
        \label{fig:VG-angle}
        \end{figure}
    
    Taking the union of the piece vertices gathered at each node sets $N^{P}(\VD)$ and $N^{P^{\prime}}(\VD)$, we can count the piece vertex set in two different ways, revealing the following fact:
    \begin{observation}\label{obs:VG-anglesum}
        Any connected component $C$ of $\VD$ satisfies
        $$\sum_{x\in N(C) \cap N^{P}(\VD)}\omega(x) = \sum_{x\in N(C) \cap N^{P^{\prime}}(\VD)}\omega(x).$$
    \end{observation}

\subsection{Relationship between Vertex and Edge Matching Diagrams}
    This section describes the relationship between paths in the edge and vertex matching diagram, which are generated by sides and their endpoint corners. 
    Formally, when a side of a piece is placed along the horizontal axis with the inside of the piece facing downward, the corner on the right is referred to as \defn{clockwise next}, while the one on the left is referred to as \defn{counterclockwise next}.
    The same definitions apply to the corners as viewed from a side.

    The following lemma describes the relationship between a path in $\ED$, starting from a side, and the $\VD$ path formed by the corners that come next along the boundary of the target shape.
    \begin{lemma}\label{lem:EGtoVG}
        Let $e$ and $e^{\prime}$ be links in $\ED$, that is, sides of pieces, and assume that they are connected by a path in $\ED$. 
        Let $v_1$ and $v_2$ be the clockwise and counterclockwise next corners of $e$, respectively. 
        Similarly, define $v^{\prime}_1$ and $v^{\prime}_2$ for $e^{\prime}$. 
        Here, $v_1$ is connected to either $v^{\prime}_1$ or $v^{\prime}_2$ by a path in $\VD$, and $v_2$ is connected to the other.
    \end{lemma}
    \begin{proof}
        Based on the assumption that $e$ and $e^{\prime}$ are connected by a path in $\ED$, each node along this path corresponds to a one-to-one pairing of two sides of the pieces. 
        Since these two sides form an a single edge of the cut graph, the vertices at both endpoints of this edge correspond to the corners of the respective sides. 
        Therefore, there exists a pair of paths in $\VD$ that shares a node between the pairs of corners to which the two next sides correspond.
    \end{proof}
    Note that which pair corresponds depends on whether the piece is flipped or not.
    When any piece is not flipped, it is possible to precisely describe which pair corresponds. 
    This will be discussed in Section~\ref{sec:anflip}.

    Next, we consider the relationship between a $\VD$ path and the $\ED$ paths formed by the next sides of the corner in the $\VD$ path. 
    Before describing this relationship, we define some terminology.
    Let us consider a paired vertex $x$ of $G^X$ such that both adjacent vertices are non-flat.
    In this case, pairs of the next sides of the corner corresponding to $x$ create degree-2 nodes in $\ED$ (see the left of Figure~\ref{fig:along-ED}). 
    On the other hand, in other cases—such as when $x$ is a corner vertex, a side vertex, a flat vertex, or a paired vertex adjacent to a flat vertex—one pair of sides does not correspond to each other and remain as degree-1 nodes, while the other pair corresponds to a degree-2 node (see the middle of Figure~\ref{fig:along-ED}). 
    In particular, when passing through a side vertex that bisects a side $e$ of a piece, the sum of the lengths of the disconnected paths is equal to the length of $e$ (see the right of Figure~\ref{fig:along-ED}).
    
    \begin{figure}[htbp]\centering
        \includegraphics[width=1\textwidth]{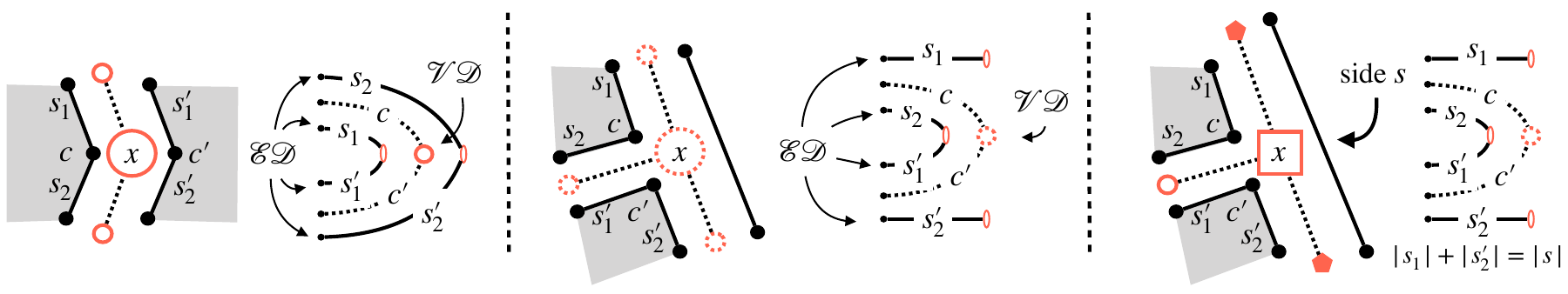}
        \caption{Examples illustrating the relationship between $\VD$ and $\ED$ paths.}
        \label{fig:along-ED}
    \end{figure}
    Based on these observations, we introduce the following definition:
    \begin{definition}
        We introduce the following notations (see Figure~\ref{fig:divide-in-two-internal}):
        \begin{itemize}
            \item $N^{X}_{\mathit{dit}}(\VD) \subset N^{X}(\VD)$ denotes the union of side vertices whose adjacent vertices are both corner vertices and of flat vertices whose adjacent vertices are both paired vertices. 
            These are referred to as \defn{divide-in-two vertices}.
            \item $N^{X}_{\mathit{bur}}(\VD) \subset N^{X}(\VD)$ denotes the set of paired vertices such that both adjacent vertices are non-flat.
            These are referred to as \defn{buried vertices}.
        \end{itemize}
        Then, a path $P = (p_1, \ldots, p_k)$ of $\VD$ is \defn{well-behaved} if $p_2, \ldots, p_{k-1} \in N^{X}_{\mathit{dit}}(\VD) \cup N^{X}_{\mathit{bur}}(\VD)$.
        \begin{figure}[htbp]\centering
            \includegraphics[width=0.2\textwidth]{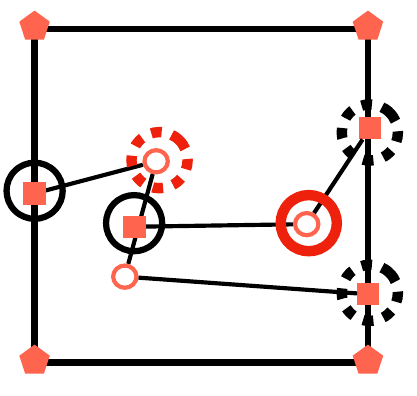}
            \caption{The points marked with a solid black circle are in $N^{X}_{\mathit{dit}}(\VD)$, while those with a dotted black circle are not. 
            Similarly, the point marked with a solid red circle is in $N^{X}_{\mathit{bur}}(\VD)$, while that with a dotted red circle is not.}
            \label{fig:divide-in-two-internal}
        \end{figure}
    \end{definition}
    
    Using the above definition, the relationship between a $\VD$ path and its adjacent $\ED$ paths can be expressed by the following lemma:
    \begin{lemma}\label{lem:alternate}
        Let $P$ be a well-behaved path in $\VD$, $v$ be an end node of $P$, and $e$ be the next side of $v$. 
        We define an \defn{along sequence} $M_1, \ldots, M_k$ of paths in $\ED$ as follows:
        \begin{itemize}
            \item Each $M_j$ is a path in $\ED$, and the end node of $M_1$ is $e$.
            \item A pair of adjacent edges—one ending $M_j$ and the other starting $M_{j+1}$—share a point in $N^{X}_{\mathit{dit}}(\VD)$ on a side $L_i$ as one of their end nodes.
        \end{itemize}
        Then, the alternating sum of $|e|, |L_1|, |L_2|, \ldots, |e^{\prime}|$ is zero where $e^{\prime}$ is the end node of $M_k$.
    \end{lemma}
        \begin{figure}[htbp]\centering
            \includegraphics[width=0.7\textwidth]{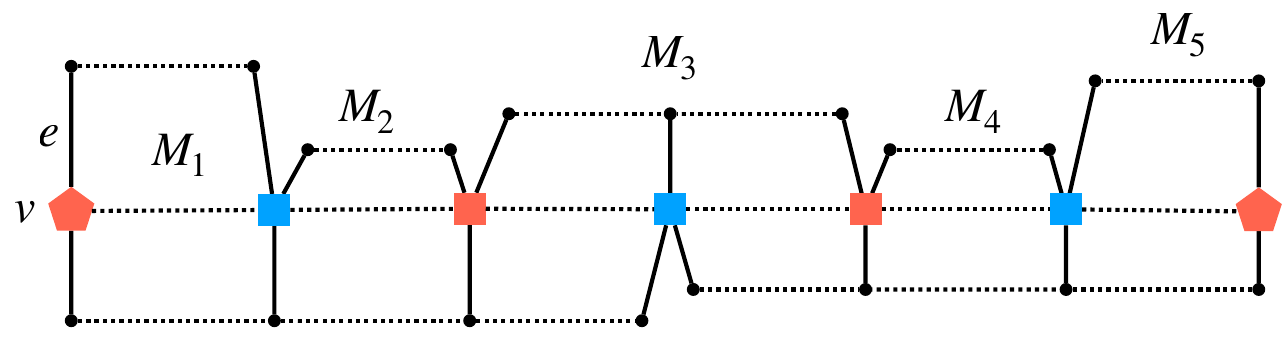}
            \caption{A well-behaved path of $\VD$ and its along sequence.}
            \label{fig:well-behaved}
        \end{figure}
    \begin{proof}
        The side $e$ and part of $L_1$ share a path in $\ED$, and according to Observation~\ref{obs:EG-length}, their lengths are equal. Moreover, the remaining part of $L_1$ and part of $L_2$ also share a path in $\ED$, and these lengths are equal as well. The equality of all these parts implies that the alternating sum of the lengths is zero.
    \end{proof}
\subsection{Lemmas when Flipping is not Allowed}\label{sec:anflip}
    When flipping of a piece is not allowed, the relationship between the connected components in Lemma~\ref{lem:EGtoVG} can be described based on the length of the paths as follows:
    \begin{lemma}\label{lem:side-component}
        We use the same notation as in Lemma~\ref{lem:EGtoVG}. 
        Let $C_1$ and $C_2$ be connected components of $\VD$, where $v_1 \in C_1$ and $v_2 \in C_2$.
        \begin{itemize}
            \item If the length of $P$ is odd, then $v^{\prime}_2 \in C_1$ and $v^{\prime}_1 \in C_2$.
            \item If the length of $P$ is even, then $v^{\prime}_1 \in C_1$ and $v^{\prime}_2 \in C_2$.
        \end{itemize}
        \begin{figure}[htbp]\centering
            \includegraphics[width=\textwidth]{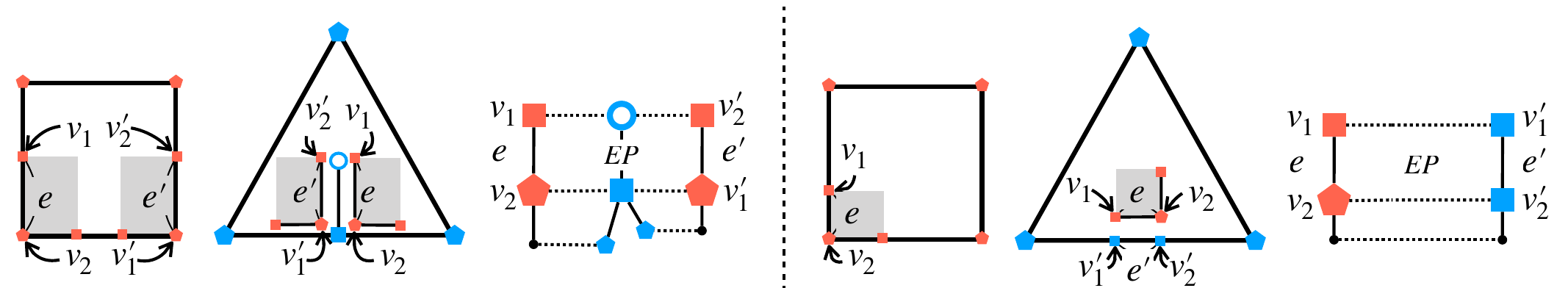}
            \caption{The two paths in $\VD$ formed by the endpoints of the path in $\ED$.}
            \label{fig:EG-length}
        \end{figure}
    \end{lemma}
    \begin{proof}
        This can be easily verified by observing that when two pieces share a side and come into contact, the clockwise-next corner on one piece corresponds to the counterclockwise-next corner on the other(Figure~\ref{fig:EG-length}).
    \end{proof}
    In addition, Lemma~\ref{lem:side-component} leads to the following lemma, which helps to simplify our proof:
    \begin{lemma}\label{lem:adjacent}
        Let $e$ and $e^{\prime}$ be boundary edges of $G^{X}$ such that they share an endpoint $v$. If $e$ and $e^{\prime}$ are connected by a path in $\ED$, the connected component of $\VD$ that includes $v$ contains a cycle.
    \end{lemma}
    \begin{proof}
        By Lemma~\ref{lem:side-component}, if $e$ and $e^{\prime}$ are connected by a path in $\ED$, we can trace a path in $\VD$ that begins and ends at $v$ (Figure~\ref{fig:adjacent}). This implies that the component containing $v$ includes a cycle.
    \end{proof}
    
        \begin{figure}[htbp]\centering
            \includegraphics[width=0.5\textwidth]{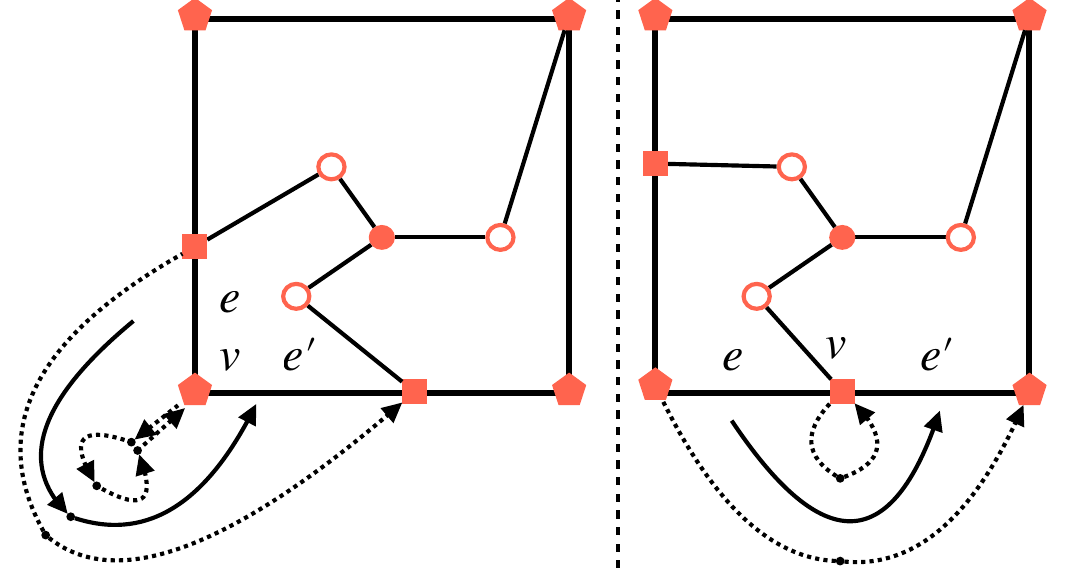}
            \caption{Illustration of cases where two edges on the boundary of $X$, either sharing a vertex of $X$ (left) or a point on an edge (right), become the endpoints of a path in $\ED$.}
            \label{fig:adjacent}
        \end{figure}
    
    In particular, we use these lemmas in combination as follows:
    \begin{lemma}\label{lem:useful}
        Let $C$ be a connected component of $\VD$ that contains a boundary node $v$, and suppose the following conditions hold:
        \begin{itemize}
            \item Both of the next sides $s$ and $s'$ of $v$ along the boundary correspond to monochromatic nodes in $\ED$.
            \item None of the next sides of any other node in $C$ correspond to monochromatic nodes in $\ED$.
        \end{itemize}
        Then $C$ forms a cycle.
    \end{lemma}
    
\section{Proof Overview and Example}\label{sec:overview-ex}
    In this section, we provide an overview of how to prove the impossibility.
    We let $P$ be a square $S$ and $P^{\prime}$ be an equilateral triangle $T$.
    We denote the side lengths of $S$ and $T$ by $\sigma$ and $\tau$, respectively, and set $\sigma = \sqrt{\sqrt{3}}$  and $\tau=2$ so that both areas are $\sqrt{3}$.
    
    First, in Section~\ref{sec:enum}, we consider the equivalence classes of the cut graph under the equivalence relation $\sim$. 
    Although this is essentially a brute-force exhaustive search, several tools are introduced to help to reduce the number of case distinctions.
    Next, in Section~\ref{sec:lemmas-tri-squ}, we perform a more detailed analysis of the properties that hold for matching diagrams when considering equilateral triangles and squares.  
    Finally, in Section~\ref{sec:individual}, using these tools, we complete the proof by deriving contradictions for each individual case.

    In the following, we present an example of the proof method by borrowing some tools in advance.
    We focus on the case where representative elements of equivalence classes $G^T$ and $G^S$ are the cut graphs shown on the left of Figure~\ref{fig:ex-howto}. 
    This pair of cut graphs corresponds to ${G^T_{\mathfrak{D}1}}$ and ${G^S_{\mathfrak{A}3}}$ in the later labeling. 
    the same proof also applies to ${G^T_{\mathfrak{D}1}}$ and ${G^S_{\mathfrak{D}7}}$.
    \begin{figure}[htbp]
        \centering
        \includegraphics[width=0.8\textwidth]{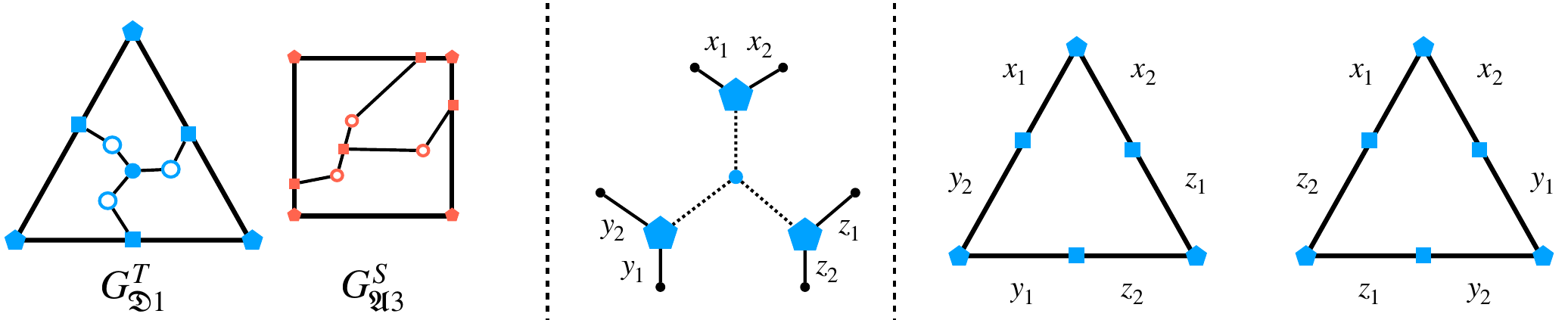}
        \caption{The pair of ${G^T_{\mathfrak{D}1}}$ and ${G^S_{\mathfrak{A}3}}$, the tree $Y$ contained in their $\VD$, and the possible combinations of side edge lengths of $T$.}
        \label{fig:ex-howto}
    \end{figure}
    \begin{lemma}
        The combination of ${G^T_{\mathfrak{D}1}}$ and ${G^S_{\mathfrak{A}3}}$ is infeasible.
    \end{lemma}
    \begin{proof}
        In this case, by Lemma~\ref{lem:struct_BandC}, which will be given later, $\VD$ contains, as a connected component, a tree $Y$ that connects each corner of $T$ and a degree-3 vertex $q$ by a well-behaved path.
        Let $v_x$, $v_y$, and $v_z$ be the three leaves of $Y$, with the adjacent sides of $v_x$ being $x_1$ and $x_2$, those of $v_y$ being $y_1$ and $y_2$, and those of $v_z$ being $z_1$ and $z_2$ (see the middle of Figure~\ref{fig:ex-howto}). 
        Here, the sides $x_1$, $x_2$, $y_1$, $y_2$, $z_1$, and $z_2$ can be grouped into three pairs, each pair forming a side of $T$. 
        Therefore, one of the following holds (see the right of Figure~\ref{fig:ex-howto}):
        \begin{itemize}
            \item \textbf{Equation 1:} $\|x_1\| + \|y_2\| = \|y_1\| + \|z_2\| = \|z_1\| + \|x_2\| = \tau$.
            \item \textbf{Equation 2:} $\|x_1\| + \|z_2\| = \|y_1\| + \|x_2\| = \|z_1\| + \|y_2\| = \tau$.
        \end{itemize}

        Let $W_x$, $W_y$, and $W_z$ be the paths in $\VD$ between each leaf and $q$.
        We consider the number of points from $N^{T}_{\mathit{dit}}(\VD)$ and $N^{S}_{\mathit{dit}}(\VD)$ included in each $W_-$. 
        Since the vertices $v_x$, $v_y$, $v_z$, and the convex vertex $q$ are all included in $G^T$, the weight of both endpoint nodes of each $W_-$ is at most $\pi$.
        By Observation~\ref{obs:VG-anglesum}, each $W_-$ contains at least one node of $N^{S}_{\mathit{dit}}(\VD)$.
        Moreover, since the weight of a divide-in-two vertex is $\pi$, the elements of $N^{T}_{\mathit{dit}}(\VD)$ appear alternately with those of $N^{S}_{\mathit{dit}}(\VD)$, and the number of $N^{T}_{\mathit{dit}}(\VD)$ contained is one less than the number of $N^{S}_{\mathit{dit}}(\VD)$.
        Therefore, since $\| N^{S}_{\mathit{dit}}(\VD) \| = 4$, the number of elements from $N^{S}_{\mathit{dit}}(\VD)$ in each path is $1$ or $2$, and the number from $N^{T}_{\mathit{dit}}(\VD)$ is $0$ or $1$.
        If $Y$ includes a point in $N^{T}_{\mathit{dit}}(\VD)$, it is a side vertex of $G^T$.

        If $Y$ includes exactly one side vertex $t$, then the only monochromatic sides connected to $Y$ are the two adjacent to $t$. 
        Therefore, the connected component including $t$ must contain a cycle by Lemma~\ref{lem:useful}, contradicting the fact that $Y$ is a tree.

        Next, we consider the case where $Y$ includes no side vertex of $T$. In this case, there are two possible structures for $Y$, as shown in Figure~\ref{fig:tree-B-1p}. 
        Depending on whether we are considering the left or right case in Figure~\ref{fig:tree-B-1p}, by Lemma~\ref{lem:alternate}, we have one of the following:
        \begin{itemize}
            \item \textbf{Equations A:} $\| x_1 \| + \| y_2 \| = \sigma$, \quad $\| y_1 \| + \| z_2 \| = \sigma$, \quad $\| z_1 \| + \| x_2 \| = \sigma$.
            \item \textbf{Equations B:} $\| x_1 \| = \| y_2 \|$, \quad $\| y_1 \| = \| z_2 \|$, \quad $\| z_1 \| + \| x_2 \| = \sigma$.
        \end{itemize}
        Here, even if $Y$ contains a flat node, by Lemma~\ref{lem:U-shape}, which will be shown later, the length of the side of the piece being divided is $\sigma$.
        
        From Equations A and Equation 1 or 2, we derive $\| x_1 \| + \| x_2 \| + \| y_1 \| + \| y_2 \| + \| z_1 \| + \| z_2 \| = 3\sigma = 3\tau,$
        which contradicts $\sigma \ne \tau$.
        From Equations B and Equation 1, we have $\| z_1 \| + \| x_2 \| = \sigma$ and $\| z_1 \| + \| x_2 \| = \tau$, which is a contradiction. 
        From Equations B and Equation 2, we have $\| z_1 \| + \| y_2 \| = \sigma$ and $\| z_1 \| + \| y_2 \| = \tau$, which is also a contradiction.
        \begin{figure}[htbp]
            \centering
            \includegraphics[width=0.5\textwidth]{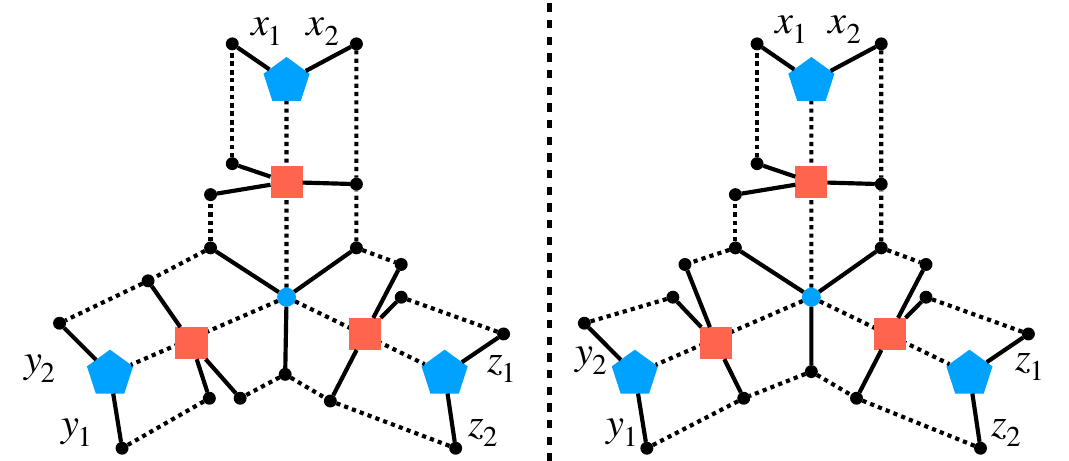}
            \caption{The possible structures for $Y$.}
            \label{fig:tree-B-1p}
        \end{figure}
    \end{proof}
    
\section{Main Proof}\label{sec:main}
    In this section, we give the proof of Theorem~\ref{thm:main}.

\subsection{Enumeration of Patterns of Cut Graphs}\label{sec:enum}
    In this section, we enumerate the equivalence classes of the equivalence relation $\sim$.
    Before the enumeration, we present several geometric lemmas.
    
\subsubsection{Geometric Lemma for Cut Graphs}
    \begin{lemma}\label{lem:geom-T-vertex}
        In any dissection of $T$ and $S$,
        no piece contains two corners of~$T$.
        Thus, the number of pieces must be at least~$3$,
        and in a three-piece dissection $P_1,P_2, P_3$,
        each $P_i$ contains exactly one corner of~$T$.
    \end{lemma}
    \begin{proof}
        Suppose a piece contains two corners of~$T$,
        whose Euclidean distance is $\tau = 2$.
        The longest distance between any two points in $S$ is $\sqrt{2\sqrt{3}}$, so no piece contains points that are farther apart than this distance.
        Since $\sqrt{2\sqrt{3}} < 2$, we obtain a contradiction.
    \end{proof}
    \begin{lemma}\label{lem:geom-S-vertex}
        Let $P_1,P_2, P_3$ be a three-piece dissection of $T$ and $S$.
        Then no $P_i$ contains two diagonal corners of $S$.
    \end{lemma}
    \begin{proof}
        If a single piece $P_i$ contains the diagonal corners of $S$, the diameter of $P_i$ is $\sqrt{2\sqrt{3}}$.
        Placing a line segment $d$ of this length $\sqrt{2\sqrt{3}}$ in $T$ requires its endpoints to reside within the region in $T$ bounded by the two blue parallel lines shown in Figure~\ref{fig:geom-S-vertex}.
        By Lemma~\ref{lem:geom-T-vertex}, piece $P_i$ must also include exactly one corner $v$ of~$T$.
        Since the line segment $d$ must be a diameter of $P_i$, $d$ must in fact have $v$ as an endpoint.
        Limited by the region, $d$ must form an angle of more than $\frac{\pi}{4}$ with an edge of~$T$ (the horizontal edge in Figure~\ref{fig:geom-S-vertex}).
        However in $S$, the diagonal has only $\frac{\pi}{4}$ of material on either side of its endpoints.
        Thus, $P_i$ cannot fully cover the corner $v$ of $T$, i.e., $v$ is in fact cut.
        This contradicts Lemma~\ref{lem:geom-T-vertex}.
    \end{proof}
    \begin{figure}[htbp]\centering
        \includegraphics[width=0.5\textwidth]{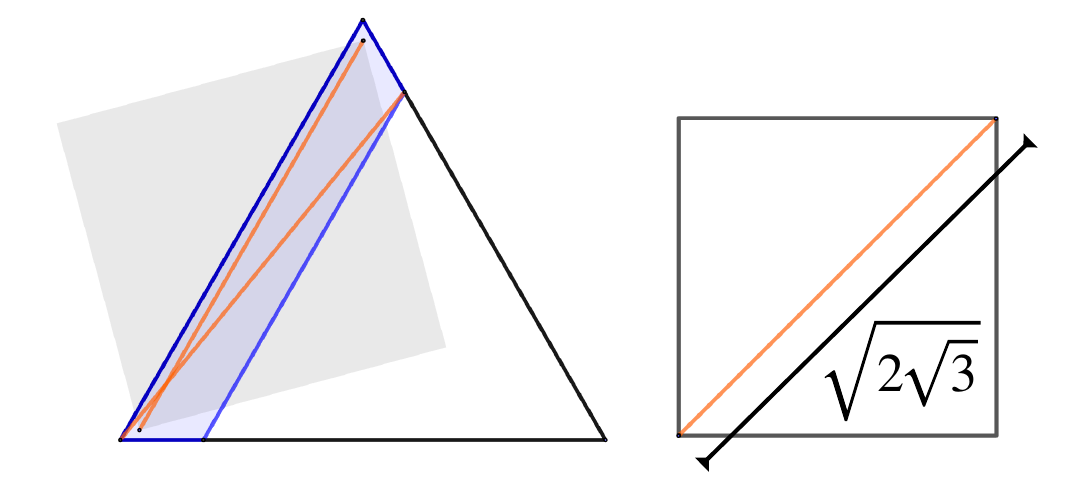}
        \caption{Ways to place a line segment of length $\sqrt{2\sqrt{3}}$, the diagonal of $S$, in~$T$.}
        \label{fig:geom-S-vertex}
    \end{figure}
    \begin{lemma}\label{lem:geom-simple}
        Let $P_1,P_2, P_3$ be a three-piece dissection of $T$ and $S$. 
        Then we can assume without loss of generality that each $P_i$ is a simple polygon.
    \end{lemma}
    \begin{proof}
        By Lemmas~\ref{lem:geom-T-vertex} and \ref{lem:geom-S-vertex}, the boundary of each polygon must be cut at least twice.  
        Thus, under the assumption of a three-piece dissection, the number of holes in a single piece is at most one, and its interior must be filled by exactly one other piece.  
        In this case (right of Figure~\ref{fig:non-simple}), we can remove the cycle to obtain a dissection with fewer pieces.  
        Similarly, if the dissection includes cut lines whose endpoints are internal (left of Figure~\ref{fig:non-simple}), we can remove these cuts while maintaining the same number of pieces.  
        Thus, in either case, we obtain a dissection with the same or smaller number of pieces that are all simple.
    \end{proof} 

    \begin{figure}[htbp]\centering
        \includegraphics[width=\textwidth]{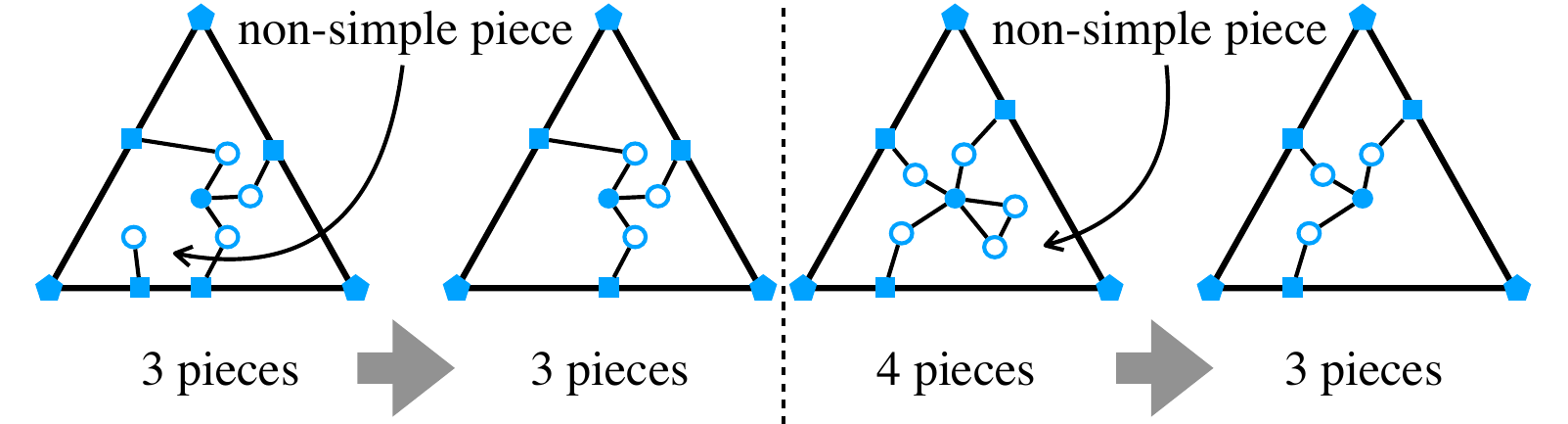}
        \caption{After removing cut lines that create non-simple polygons, the obtained dissection has the same or fewer number of pieces.}
        \label{fig:non-simple}
    \end{figure}
    
\subsubsection{Procedure for Enumerating Cut Graphs}
    We classify the possible equivalence classes of the relation $\sim$ using Lemmas~\ref{lem:geom-T-vertex}, \ref{lem:geom-S-vertex}, and \ref{lem:geom-simple}.

    We begin by abstracting the boundary of the target shape as a disk, ignoring differences in vertex types.  
    First, consider the case of partitioning the disk into two pieces.  
    Since each piece must be a simple polygon (by Lemma~\ref{lem:geom-simple}), the boundary must be cut at least once.  
    As shown on the left of Figure~\ref{fig:abstract-enum}, there is a unique way to perform such a cut: connect two points on the circular boundary with a cut line~$C$.  
    The resulting boundary of the pieces consists of three types of segments: parts from the cut line~$C$, parts from the original boundary, and two switching points between them.

    To obtain three pieces, we make a second cut $C^{\prime}$ on one of the two pieces, again connecting two boundary points.  
    As shown on the right of Figure~\ref{fig:abstract-enum}, there are exactly six distinct ways to make such a second cut.  
    These yield six topologically distinct configurations, which we enumerate explicitly in Figure~\ref{fig:abstract-enum} and label as \defn{Type~$\mathfrak{A}$} through \defn{Type~$\mathfrak{F}$}.
    \begin{figure}[htbp]\centering
        \includegraphics[width=\textwidth]{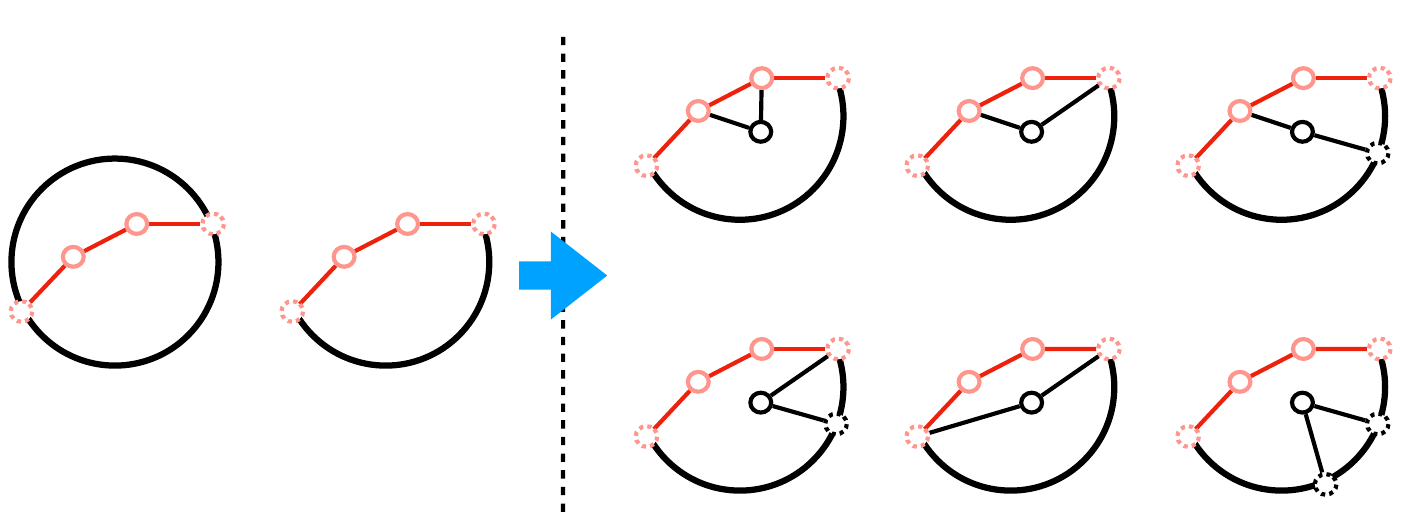}
        \caption{Left: The unique configuration for the first cut of the disk and the resulting boundary structure. Right: The six possible configurations for the second cut.}
        \label{fig:abstract-enum-pre}
    \end{figure}
    \begin{figure}[htbp]\centering
        \includegraphics[width=\textwidth]{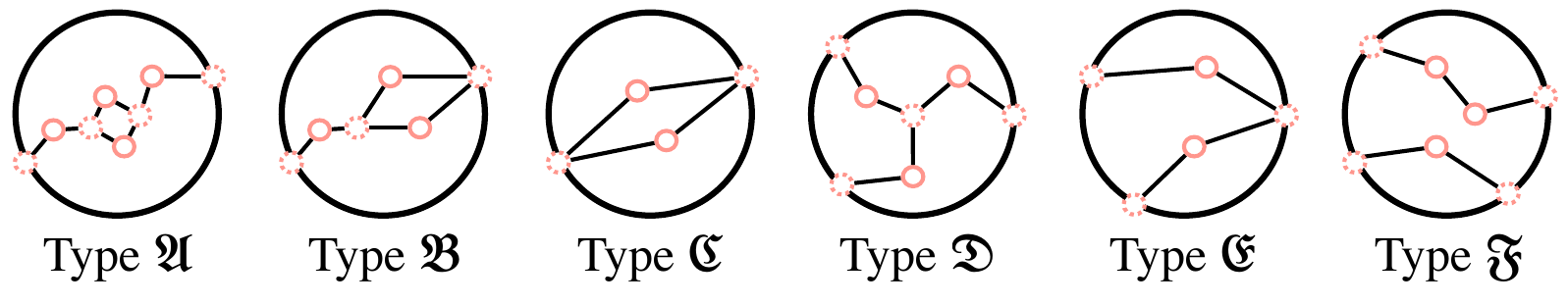}
        \caption{The topological classification of cut graphs when the boundary is abstracted.}
        \label{fig:abstract-enum}
    \end{figure}

    We now apply this abstract classification to the actual dissection of the target shapes $S$ and $T$.  
    To do so, we must account for how the endpoints of each cut line attach to the corners or sides of the target shape, as well as the presence of vertices of degree three or higher.
    Using Lemma~\ref{lem:geom-T-vertex}, we enumerate the equivalence classes of $G^T$ as shown in Figure~\ref{fig:class_T}, denoting them as $G^T_i$.
    \begin{figure}[htbp]\centering
        \includegraphics[width=\textwidth]{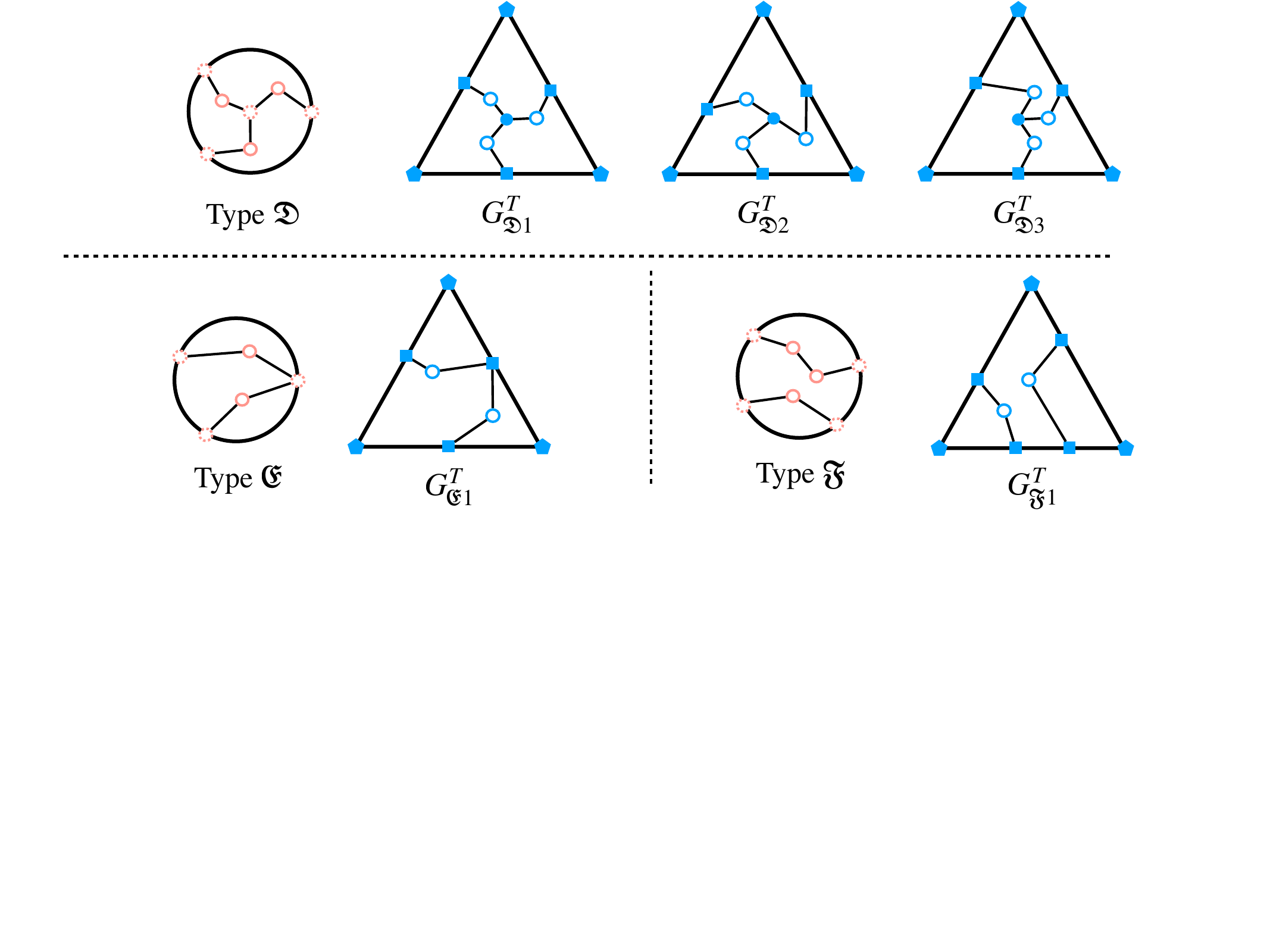}
        \caption{The possible equivalence classes of $G^T$. Open circles represent possible paths of degree-2 vertices from subdivision. $G^T_{\mathfrak{D}1}$, $G^T_{\mathfrak{D}2}$, and $G^T_{\mathfrak{D}3}$ have a degree-3 vertex of type 4, 6 (flat), and 5, respectively.}
        \label{fig:class_T}
    \end{figure}
    Similarly, applying Lemma~\ref{lem:geom-S-vertex}, we enumerate the equivalence classes of $G^S$ as shown in Figures~\ref{fig:class_S-1} and \ref{fig:class_S-2}, denoted as $G^S_j$.
    \begin{figure}[htbp]\centering
        \includegraphics[width=\textwidth]{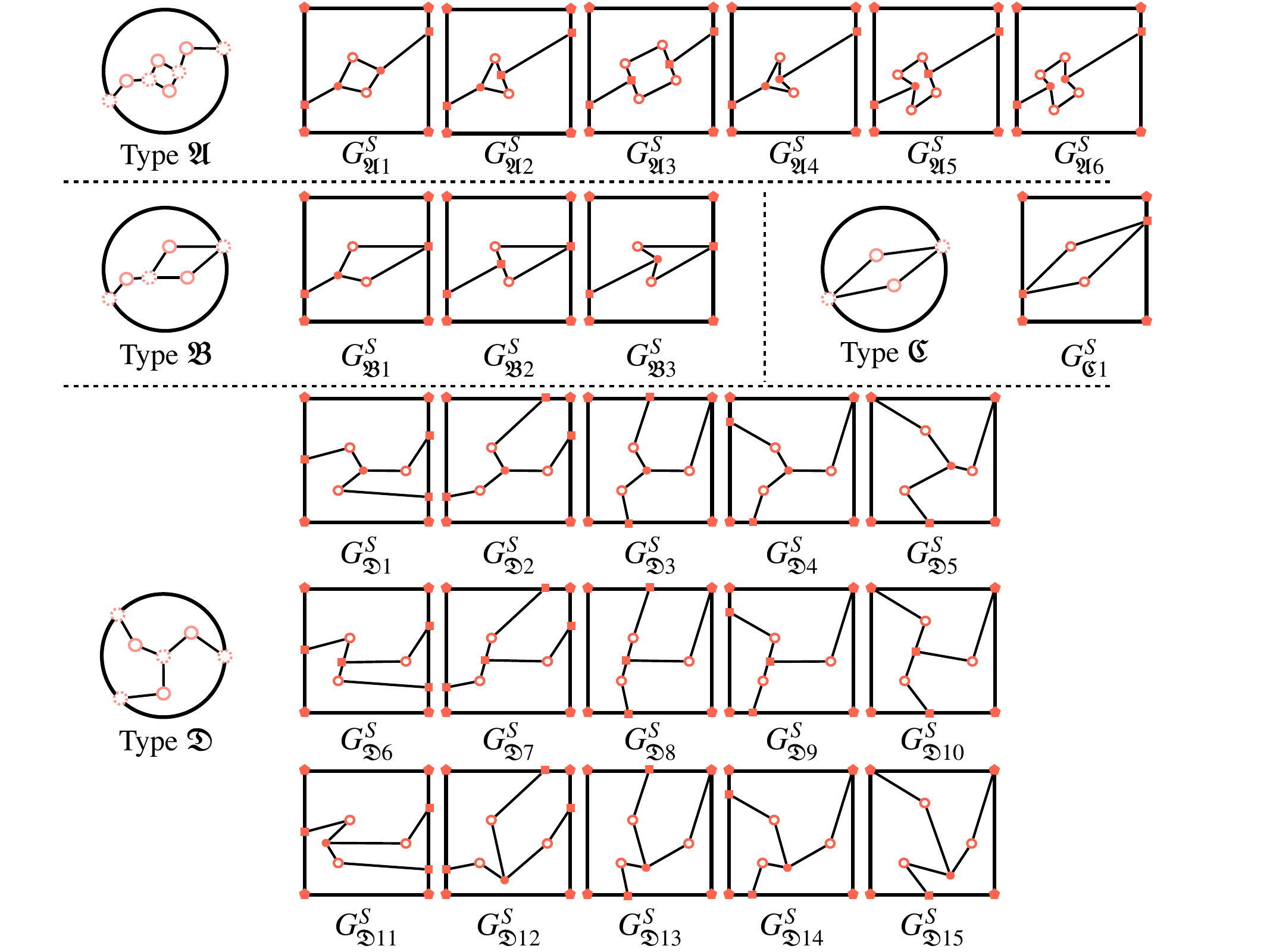}
        \caption{The possible equivalence classes of $G^S$ for \defn{Type~$\mathfrak{A}$}, \defn{Type~$\mathfrak{B}$}, \defn{Type~$\mathfrak{C}$}, and \defn{Type~$\mathfrak{D}$}.}
        \label{fig:class_S-1}
    \end{figure}
    \begin{figure}[htbp]\centering
        \includegraphics[width=\textwidth]{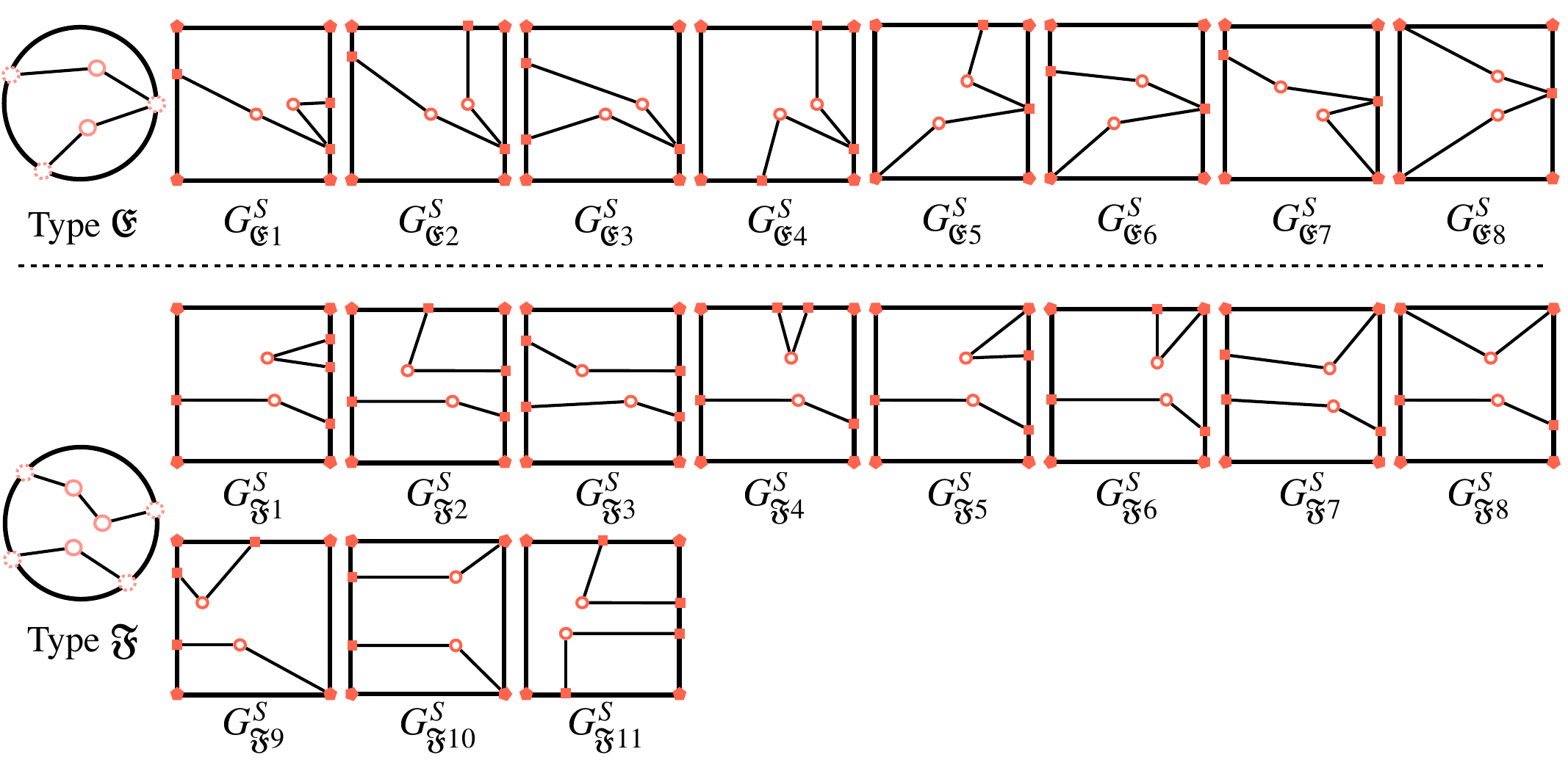}
        \caption{The possible equivalence classes of $G^S$ for \defn{Type~$\mathfrak{E}$} and \defn{Type~$\mathfrak{F}$}}
        \label{fig:class_S-2}
    \end{figure}

    To identify feasible pairings between $G^S_{-}$ and $G^T_{-}$, we introduce the following angular invariant:
    \begin{definition}
        For a graph $G^X$ and an angle $\theta < \pi$, the \defn{$\theta$-diff} is defined as the number of times $\theta$ appears minus the number of times $2\pi - \theta$ appears as an interior angle.  
        We define \defn{cc-diff} as the sum of all $\theta$-diffs for $\theta < \pi$, and \defn{tri-diff} as the difference $\pi/3$-diff $-\, 2\pi/3$-diff.
    \end{definition}
    Since $\theta$-diff is preserved under subdivision, the cc-diff of each $G^T_i$ and $G^S_j$ can be determined by counting convex minus reflex angles at solid vertices in Figures~\ref{fig:class_T}--\ref{fig:class_S-2}:
    \begin{itemize}
        \item The value of cc-diff is 
            $12$ for $G^T_{\mathfrak{D}1}$, 
            $11$ for each of $\{G^T_{\mathfrak{D}2}, G^T_{\mathfrak{F}1}\}$, and
            $10$ for each of $\{G^T_{\mathfrak{D}3}, G^T_{\mathfrak{E}1}\}$.
        \item The value of cc-diff is 
        \begin{itemize}
            \item $14$ for $G^S_{\mathfrak{A}1}$, 
            \item $13$ for each of $\{G^S_{\mathfrak{A}2}, G^S_{\mathfrak{D}1}, G^S_{\mathfrak{D}2}\}$,
            \item $12$ for each of $\{G^S_{\mathfrak{A}3}, G^S_{\mathfrak{A}4}, G^S_{\mathfrak{B}1}, G^S_{\mathfrak{D}3}, G^S_{\mathfrak{D}4}, G^S_{\mathfrak{D}6}, G^S_{\mathfrak{D}7}, G^S_{\mathfrak{F}1}, G^S_{\mathfrak{F}2}, G^S_{\mathfrak{F}3}, G^S_{\mathfrak{F}4}, G^S_{\mathfrak{F}11}\}$,
            \item $11$ for each of $\{G^S_{\mathfrak{A}5}, G^S_{\mathfrak{B}2}, G^S_{\mathfrak{D}5}, G^S_{\mathfrak{D}8}, G^S_{\mathfrak{D}9}, G^S_{\mathfrak{D}11}, G^S_{\mathfrak{D}12}, G^S_{\mathfrak{E}1}, G^S_{\mathfrak{E}2}, G^S_{\mathfrak{E}3}, G^S_{\mathfrak{E}4}, G^S_{\mathfrak{F}5}, G^S_{\mathfrak{F}6}, G^S_{\mathfrak{F}7} , G^S_{\mathfrak{F}9}\}$,
            \item $10$ for each of $\{G^S_{\mathfrak{A}6}, G^S_{\mathfrak{B}3}, G^S_{\mathfrak{C}1}, G^S_{\mathfrak{D}10}, G^S_{\mathfrak{D}13}, G^S_{\mathfrak{D}14}, G^S_{\mathfrak{E}5}, G^S_{\mathfrak{E}6}, G^S_{\mathfrak{E}7}, G^S_{\mathfrak{F}8}, G^S_{\mathfrak{F}10}\}$, and
            \item $9$ for each of $\{G^S_{\mathfrak{D}15}, G^S_{\mathfrak{E}8}\}$.
        \end{itemize}
    \end{itemize}
    Moreover, tri-diff can be computed as at least $3$ for each of $\{G^T_{\mathfrak{D}2}, G^T_{\mathfrak{F}1}\}$ and less than $3$ for each of $\{G^S_{\mathfrak{A}5}, \allowbreak G^S_{\mathfrak{D}8}, \allowbreak G^S_{\mathfrak{D}9}, \allowbreak G^S_{\mathfrak{D}11}, \allowbreak G^S_{\mathfrak{D}12}, \allowbreak G^S_{\mathfrak{F}5}, G^S_{\mathfrak{F}6}, \allowbreak G^S_{\mathfrak{F}7}, \allowbreak G^S_{\mathfrak{F}9}\}$.

    In order to generate the same set of pieces by $G^S$ and $G^T$, $\theta$-diff must match between them for every $\theta < \pi$.
    Therefore, only the following combinations can be feasible, where $\{G^T_{i}, \ldots\} \times \{G^S_{j}, \ldots\}$ denotes the set of all pairwise combinations of $G^T_i$ and $G^S_j$:
    \begin{itemize}
        \item $\{G^T_{\mathfrak{D}1}\} \times \{G^S_{\mathfrak{A}3}, G^S_{\mathfrak{A}4}, G^S_{\mathfrak{B}1}, G^S_{\mathfrak{D}3}, G^S_{\mathfrak{D}4}, G^S_{\mathfrak{D}6}, G^S_{\mathfrak{D}7}, G^S_{\mathfrak{F}1}, G^S_{\mathfrak{F}2}, G^S_{\mathfrak{F}3}, G^S_{\mathfrak{F}4}, G^S_{\mathfrak{F}11}\}$,
        \item $\{G^T_{\mathfrak{D}2}, G^T_{\mathfrak{F}1}\} \times \{G^S_{\mathfrak{B}2}, G^S_{\mathfrak{D}5}, G^S_{\mathfrak{E}1}, G^S_{\mathfrak{E}2}, G^S_{\mathfrak{E}3}, G^S_{\mathfrak{E}4}\}$,
        \item $\{G^T_{\mathfrak{D}3}, G^T_{\mathfrak{E}1}\} \times \{G^S_{\mathfrak{A}6}, G^S_{\mathfrak{B}3}, G^S_{\mathfrak{C}1}, G^S_{\mathfrak{D}10}, G^S_{\mathfrak{D}13}, G^S_{\mathfrak{D}14}, G^S_{\mathfrak{E}5}, G^S_{\mathfrak{E}6}, G^S_{\mathfrak{E}7}, G^S_{\mathfrak{F}8}, G^S_{\mathfrak{F}10}\}$.
    \end{itemize}

\subsection{Lemmas for Equilateral Triangle and Square}\label{sec:lemmas-tri-squ}
    In this section, we present two lemmas that are useful in proving properties related to the dissection between an equilateral triangle and a square.

    The first lemma constrains on the structure of the connected components of $\VD$ and can be described as follows:
    \begin{lemma}\label{lem:VG-component}
        Any connected component of $\VD$ includes either all corner vertices of $T$ or none, and includes an even number of corner vertices of $S$.
    \end{lemma}
    \begin{proof}
        Except for the corner vertices of $S$ and $T$, the weight of any vertex is a multiple of $\pi$ (see Figure~\ref{fig:VG-component}).
        Therefore, if the conditions in the statement do not hold, the sums on both sides will not match.
        It contradicts to Observation \ref{obs:VG-anglesum}.
    \end{proof}

        \begin{figure}[htbp]\centering
            \includegraphics[width=\textwidth]{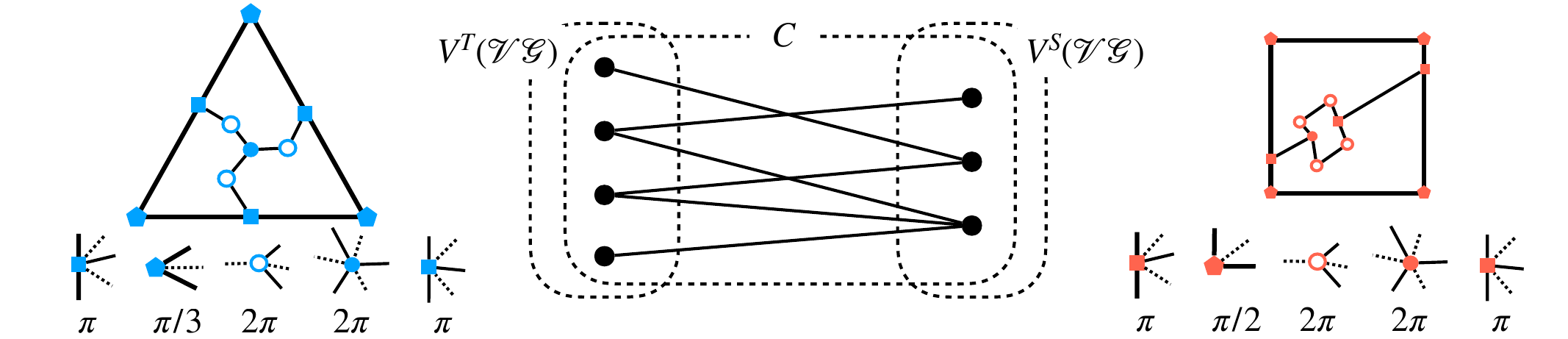}
            \caption{A connected component of $\VD$ and the sum of interior angles within it.}
            \label{fig:VG-component}
        \end{figure}
    
    The second lemma is used to handle well-behavior properties.
    When dealing with the 3-piece dissection between an equilateral triangle and a square, a path ceases to be well-behaved when it passes through a corner vertex, a vertex adjacent to a flat vertex, or a side vertex located on a trisected edge of the target shape.
    Here, in addition to the node in $\ED$ added for an edge containing a flat node, we also define the node corresponding to the middle of a trisected side as a \defn{trisected-mid} node.
    Thus, to verify that a path is well-behaved, it is sufficient to check that it does not pass through a corner vertex and does not go through the endpoints of a flat node or a trisected-mid node.
    Regarding the behavior of flat nodes and trisected-mid nodes in $\ED$, the following holds:
    \begin{definition}
        Let $e$ be a boundary edge of $G^S$ and $v,v^{\prime}$ be the adjacent vertices of $x$. 
        $(v, e, v^{\prime})$ is called a \defn{U-shaped boundary} if both $v$ and $v^{\prime}$ are a corner vertex of degree 1 in~$\VD$.
    \end{definition}
     \begin{lemma}\label{lem:U-shape}
        Assume that $(v, e, v^{\prime})$ forms a U-shaped boundary and that the connected components containing $v$ and $v^{\prime}$ (which may be the same) are both simple paths. Then, the other end node $e^{\prime}$ of the path in $\ED$ with $e$ as an end node is either a flat node or a trisected-mid node in $\ED$.
    \end{lemma}
    \begin{proof}
        Let $W = (v, v_1, \ldots, v_k)$ and $W^{\prime} = (v^{\prime}, v^{\prime}_1, \ldots, v^{\prime}_{k^{\prime}})$ be the paths in $\VD$ beginning from $v$ and $v^{\prime}$, respectively. 
        
        First, we show that for each of the sequences $v_1, \ldots, v_{k-1}$ and $v^{\prime}_1, \ldots, v^{\prime}_{k^{\prime}-1}$, at least one element in each sequence is either a flat vertex or a boundary vertex of $S$ or $T$. 
        It suffices to show this for $v_1, \ldots, v_{k-1}$.
        Suppose, for contradiction, that all $v_1, \ldots, v_{k-1}$ are paired  vertices. 
        Then, by Lemma~\ref{obs:VG-angle}, $\angle{v_k}$ would be either $\pi/2$ or $3\pi/2$, depending on whether $v_k \in N^{T}(\VD)$ or $v_k \in N^{S}(\VD)$. 
        However, since $v_k$ is an endpoint of the path, it must be a corner vertex of $X \in \{S, T\}$. In the case of $T$, the internal angle at $v_k$ would be $\pi/3$, and in the case of $S$, it would be $\pi/2$, leading to a contradiction.
    
        Therefore, both $v_1, \ldots, v_{k-1}$ and $v^{\prime}_1, \ldots, v^{\prime}_{k^{\prime}-1}$ must contain at least one flat vertex or boundary vertex of $S$ or $T$. Let $s$ and $s^{\prime}$ be the first such points encountered along $W$ and $W^{\prime}$ from $v$ and $v^{\prime}$, respectively. By Lemma~\ref{obs:VG-angle}, both $s$ and $s^{\prime}$ must belong to $N^{T}(\VD)$.
        Now, assume for contradiction that the statement of the lemma does not hold. 
        Suppose that $e^{\prime}$ is neither a flat node nor trisected-mid node. 
        Then, $e^{\prime}$ cannot have $s$ (or $s^{\prime}$) as its adjacent side. 
        If it did, the other endpoint of $e^{\prime}$ would be a corner $t$ of $T$, and $t$ and $s^{\prime}$ (or $s$) would be connected in $\VD$, which contradicts the fact that $W^{\prime}$ (or $W$) is a path and Lemma~\ref{lem:VG-component}.
        Therefore, at both $s$ and $s^{\prime}$, the path in $\ED$ starting from $e$ has degree-2 and corresponds to cut lines of length $\sigma$ perpendicular to a side of $S$ at both $s$ and $s^{\prime}$. 
        If the cut lines attached to $s$ and $s^{\prime}$ are the same, this would imply that $T$ has a pair of parallel sides on its boundary (right of Figure~\ref{fig:U-shape}), leading to a contradiction. 
        On the other hand, if the cut lines attached to $s$ and $s^{\prime}$ are different, this implies the existence of a pair of cut lines of length $\sigma$ orthogonal to the sides of $T$, which intersect within $T$ (left of Figure~\ref{fig:U-shape}), leading to another contradiction.
    
        In the case where $W$ and $W^{\prime}$ are the same path, it is possible that $s$ and $s^{\prime}$ are the same vertex. However, this would mean that the endpoint of the path in $\ED$ that starts at $e$ would be $e$ itself, contradicting the definition of a path.
    \end{proof}
    
        \begin{figure}[htbp]
            \centering
            \includegraphics[width=\textwidth]{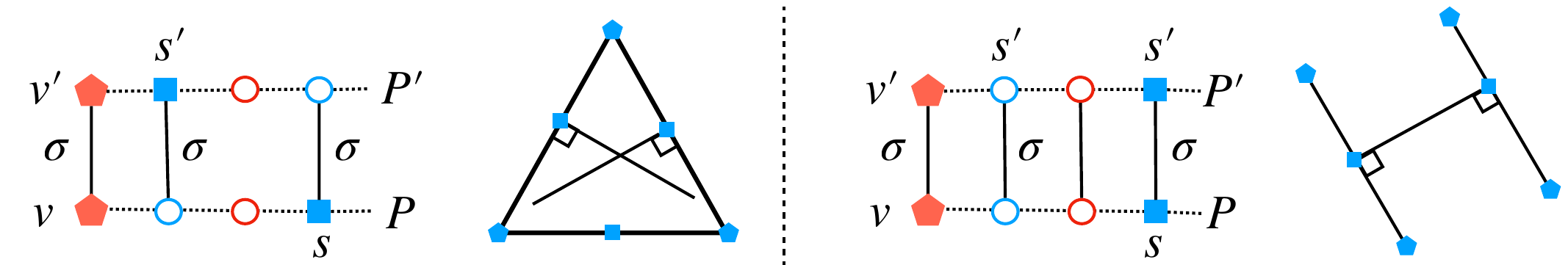}
            \caption{Illustration of when the cut lines attached to $s$ and $s^{\prime}$ are the same (left) and different (right).}
            \label{fig:U-shape}
        \end{figure}
    
\subsection{Proofs for Individual Cases}\label{sec:individual}
    In this section, we apply the approach to eliminate the remaining cases narrowed down in Section~\ref{sec:enum}.
    Rather than deriving a classification from abstract principles, we arrived at a practical case division through extensive trial-and-error analysis of all surviving cut graphs.
    Based on recurring structural features, we group the configurations into three representative types, which we refer to as Cases~A, B, and~C.
    Each of these captures a common geometric obstruction, and together they account for all infeasible instances identified in the enumeration.
    We begin by formulating Lemma~\ref{lem:struct_BandC} that justifies this trichotomy.
    \begin{lemma}\label{lem:struct_BandC}
        We define Cases A, B, and C as follows:
        \begin{itemize}
            \item \defn{Case A}: 
            $\{G^T_{\mathfrak{D}1}\} \times \{G^S_{\mathfrak{D}6}, G^S_{\mathfrak{F}1}, G^S_{\mathfrak{F}2}, G^S_{\mathfrak{F}4}, G^S_{\mathfrak{F}11}, G^S_{\mathfrak{F}3}\}$,
            $\{G^T_{\mathfrak{D}2}, G^T_{\mathfrak{F}1}\} \times \{G^S_{\mathfrak{B}2}. G^S_{\mathfrak{E}1}, G^S_{\mathfrak{E}3}\}$, 
            $\{G^T_{\mathfrak{D}3}, G^T_{\mathfrak{E}1}\} \times \{G^S_{\mathfrak{F}8}\}$
            \item \defn{Case B}: 
            $\{G^T_{\mathfrak{D}1}\} \times \{G^S_{\mathfrak{A}3}, G^S_{\mathfrak{D}7}\}$, 
            $\{G^T_{\mathfrak{D}2}, G^T_{\mathfrak{F}1}\} \times \{G^S_{\mathfrak{E}2}, G^S_{\mathfrak{E}4}\}$
            \item \defn{Case C}: 
            $\{G^T_{\mathfrak{D}1}\} \times \{G^S_{\mathfrak{A}4},G^S_{\mathfrak{B}1}, G^S_{\mathfrak{D}3}, G^S_{\mathfrak{D}4}\}$,
            $\{G^T_{\mathfrak{D}2},G^T_{\mathfrak{F}1}\} \times \{G^S_{\mathfrak{D}5}\}$, 
            $\{G^T_{\mathfrak{D}3},G^T_{\mathfrak{E}1}\} \times \{G^S_{\mathfrak{A}6}, \allowbreak G^S_{\mathfrak{B}3}, \allowbreak G^S_{\mathfrak{C}1}, \allowbreak G^S_{\mathfrak{D}10}, \allowbreak G^S_{\mathfrak{D}13}, \allowbreak G^S_{\mathfrak{D}14}, \allowbreak G^S_{\mathfrak{E}5}, G^S_{\mathfrak{E}6}, \allowbreak G^S_{\mathfrak{E}7}, G^S_{\mathfrak{F}10}\}$
        \end{itemize}
        \begin{figure}[htbp]\centering
            \includegraphics[width=\textwidth]{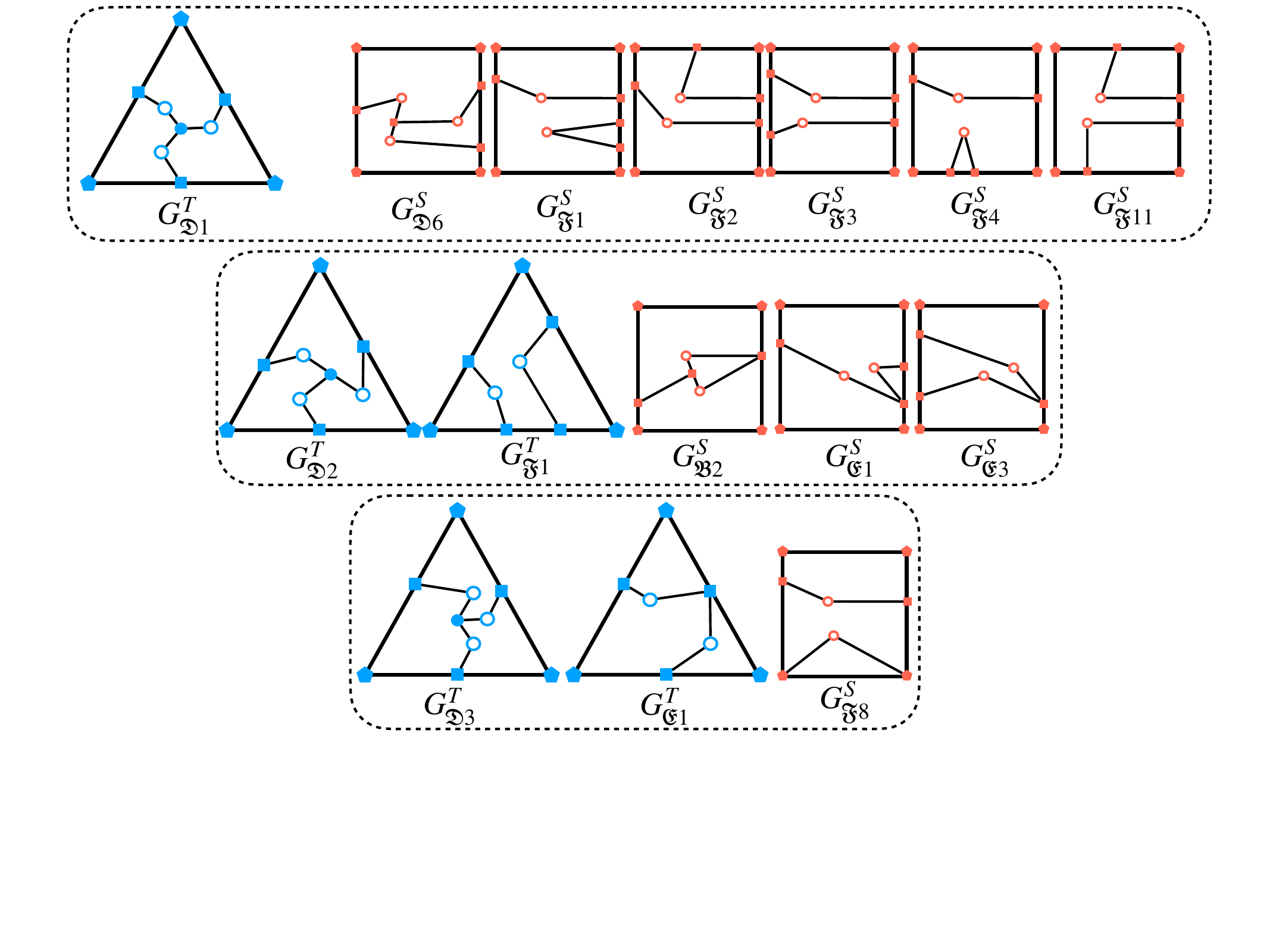}
            \caption{The combinations of cut line graphs in Case A.}
            \label{fig:CaseA}
        \end{figure}
        \begin{figure}[htbp]\centering
            \includegraphics[width=\textwidth]{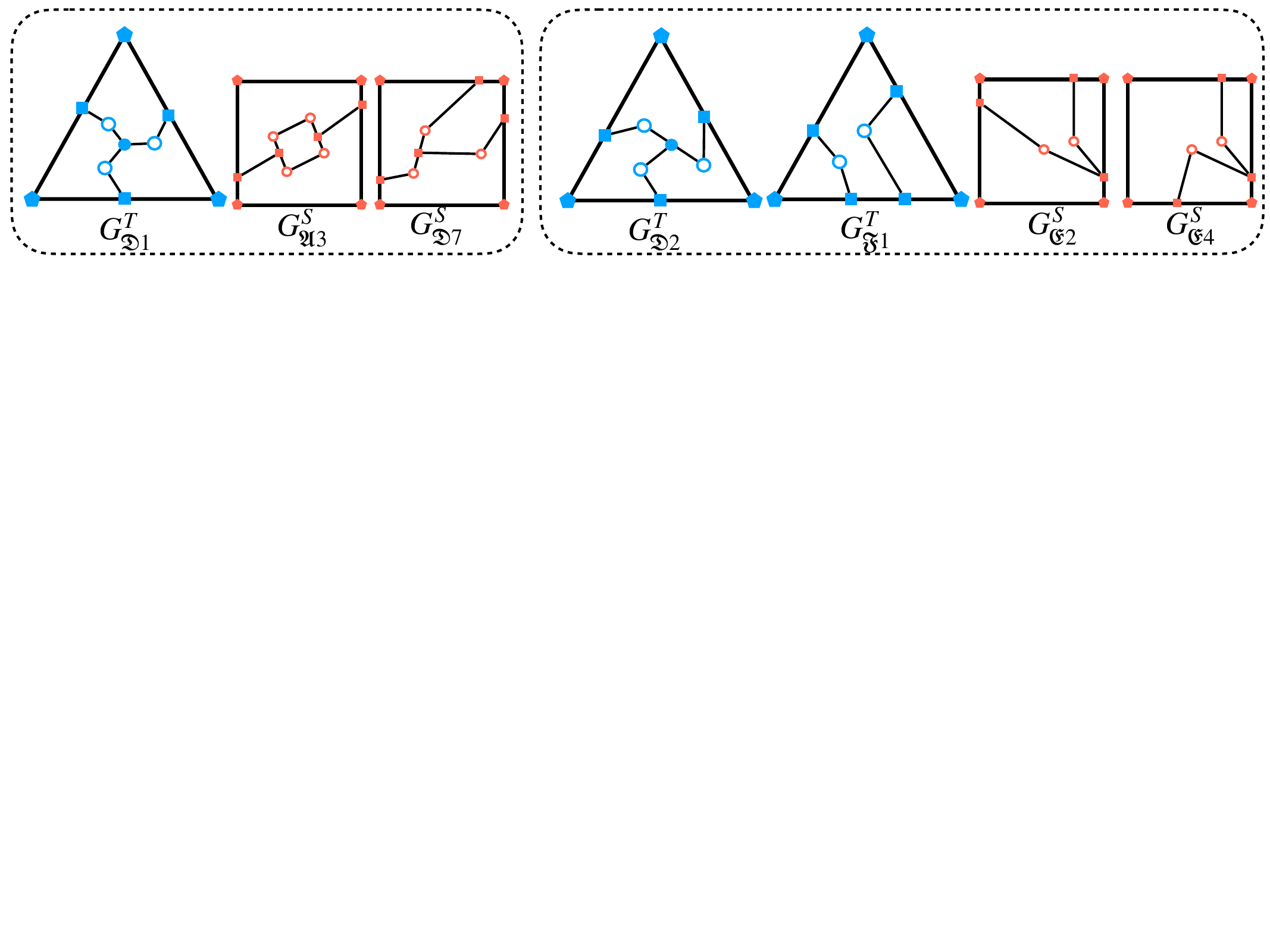}
            \caption{The combinations of cut line graphs in Case B.}
            \label{fig:CaseB}
        \end{figure}
        \begin{figure}[htbp]\centering
            \includegraphics[width=\textwidth]{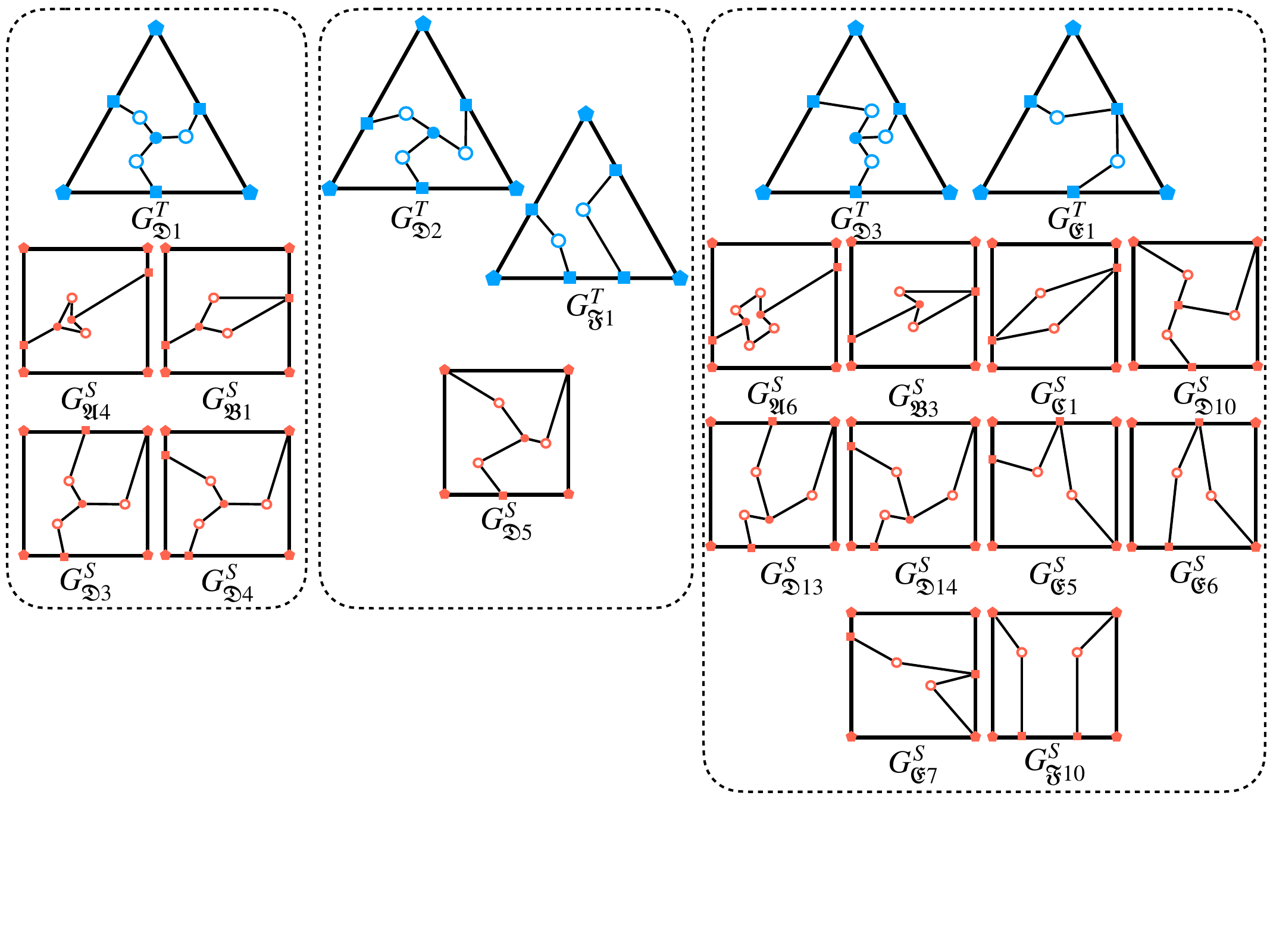}
            \caption{The combinations of cut line graphs in Case C.}
            \label{fig:CaseC}
        \end{figure}

        In Case A (Figure~\ref{fig:CaseA}), the number of flat edges and the middle of trisected edges of $S$ or $T$ is less than the number of U-shaped boundaries. 
        In Case B (Figure~\ref{fig:CaseB}), there exists a \defn{well-behaved tree} $Y$ of $\ED$ that connects a degree-$3$ vertex and each corners of $T$ by a well-behaved path.
        In Case C (Figure~\ref{fig:CaseC}), there exists a well-behaved path $W$ that connects vertices of $S$.
    \end{lemma}
    \begin{proof}
        In Case A, the statement clearly holds.
    
        In Case B, by Lemma~\ref{lem:VG-component}, the connected components of $\VD$ consist of a tree $Y$ connecting the corners of $T$ and paths $W_1$ and $W_2$ connecting the corners of $S$, since there is a single vertex of degree 3 of $G^S$ and $G^T$. 
        According to Lemma~\ref{lem:U-shape}, the paths in $\ED$ starting from the edges of the U-shaped boundary have either a flat node or a trisected-mid node as an endpoint. 
        Thus, the adjacent vertices belong to either $W_1$ or $W_2$. 
        Consequently, all three paths composing $Y$ are well-behaved.
    
        In Case C, there are at most three vertices of degree 3 in $\VD$, and the number of leaves in a connected component is at most five. 
        Therefore, by Lemma~\ref{lem:VG-component}, $\VD$ must include a connected component that forms a path $W$ connecting a pair of corners of $S$. 
        It is clear that $W$ is well-behaved in the combinations $\{G^T_{\mathfrak{D}1}\} \times \{G^S_{\mathfrak{A}4}, G^S_{\mathfrak{B}1}, G^S_{\mathfrak{D}3}, G^S_{\mathfrak{D}4}\}$ and $\{G^T_{\mathfrak{D}3}, G^T_{\mathfrak{E}1}\} \times \{G^S_{\mathfrak{A}6}, G^S_{\mathfrak{B}3}, G^S_{\mathfrak{C}1}, G^S_{\mathfrak{D}13}, G^S_{\mathfrak{D}14}, G^S_{\mathfrak{E}5}, G^S_{\mathfrak{E}6}, G^S_{\mathfrak{E}7}\}$, since by Observation~\ref{obs:VG-angle}, $W$ does not pass through corner vertices of $S$ with degree greater than two. 
        Now, we consider the remaining cases of $\{G^T_{\mathfrak{D}2}, G^T_{\mathfrak{F}1}\} \times \{G^S_{\mathfrak{D}5}\}$ and $\{G^T_{\mathfrak{D}3}, G^T_{\mathfrak{E}1}\} \times \{G^S_{\mathfrak{D}10}, G^S_{\mathfrak{F}10}\}$. 
        Let $e$ be the side of $S$ whose both corners have degree greater than 2 in $\VD$. 
        Since $e$ is monochromatic, any edge $e^{\prime}$ sharing a path with $e$ in $\ED$ is also monochromatic. 
        If $e^{\prime}$ is a flat node or a trisected-mid node, $W$ does not pass through these endpoints, making it well-behaved. 
        If $e^{\prime}$ is an edge bisecting an side of $T$, then any edge $e^{\prime\prime}$ sharing this edge with $e^{\prime}$ is also monochromatic. 
        By following paths in $\ED$ recursively, we eventually reach a flat node or a trisected-mid node, which is also monochromatic. 
        Thus, $W$ does not pass through these vertices, making it well-behaved.
    \end{proof}

    The structural property established in the above lemma allows us to rule out each of the remaining cases. 
    We now examine them one by one.
    \begin{lemma}
        Case A is infeasible.
    \end{lemma}
    \begin{proof}
        Each Case A pattern contains a single vertex of degree 3, and the three corners of $T$ are connected by a tree $Y$ that includes this degree-3 vertex. 
        Therefore, the corners of $S$ on the U-shaped boundary $(v, e, v^{\prime})$ are both included in paths of $\VD$.
    
        By Lemma~\ref{lem:U-shape}, the other endpoint of the path in $\ED$ starting from $e$ must be either a flat node or a trisected-mid node.
        Furthermore, by Lemma~\ref{lem:struct_BandC}, the number of such edges is smaller than the number of U-shaped boundaries.
        This leads to a contradiction, proving that Case A is infeasible.
    \end{proof}
    \begin{lemma}
        Case B is infeasible.
    \end{lemma}
    \begin{proof}
        In this case, there exists a well-behaved tree $Y$ by Lemma~\ref{lem:struct_BandC}.
        Let $W_x$, $W_y$, and $W_z$ be the paths in $\VD$ between each leaf and $q$.
        Let $v_x$, $v_y$, and $v_z$ be the three leaves of $Y$ with the sides next to $v_x$ being $x_1$ and $x_2$, those connected to $v_y$ being $y_1$ and $y_2$, and those connected to $v_z$ being $z_1$ and $z_2$.
        Here, the sides $x_1$, $x_2$, $y_1$, $y_2$, $z_1$, and $z_2$ can be grouped into three pairs, each pair forming a side of $T$. 
        
        Therefore, one of the following holds (see the left of Figure~\ref{fig:ex-howto}):
        \begin{itemize}
            \item \textbf{Equation 1:} $x_1 + y_2 = y_1 + z_2 = z_1 + x_2 = \tau$.
            \item \textbf{Equation 2:} $x_1 + z_2 = y_1 + x_2 = z_1 + y_2 = \tau$.
        \end{itemize}
        
        We now consider two cases based on the location of the degree-3 vertex $q$ of $Y$: in Case B-1, $q$ belongs to $G^T$, while in Case B-2, it belongs to $G^S$.

        \textbf{\boldmath Case B-1: The degree-3 vertex $q$ is in $G^T$:}  
        This proof is a repetition of Section~\ref{sec:overview-ex}. 
        We analyze the number of points from $N^{T}_{\mathit{dit}}(\VD)$ and $N^{S}_{\mathit{dit}}(\VD)$ included in each $W_-$.  
        Since the vertices $v_x$, $v_y$, $v_z$ and the convex vertex $q$ are all included in $G^T$, the weight of both endpoint nodes of each $W_-$ is at most $\pi$.  
        By Observation~\ref{obs:VG-anglesum}, each $W_-$ contains at least one node of $N^{S}_{\mathit{dit}}(\VD)$, and the number of $N^{T}_{\mathit{dit}}(\VD)$ is one less than that of $N^{S}_{\mathit{dit}}(\VD)$.  
        Since $\| N^{S}_{\mathit{dit}}(\VD) \| \leq 4$, each path contains at most two elements from $N^{S}_{\mathit{dit}}(\VD)$ and at most one from $N^{T}_{\mathit{dit}}(\VD)$.  
        If $Y$ includes a point in $N^{T}_{\mathit{dit}}(\VD)$, it must be a side vertex of $G^T$.  

        We now analyze how many side vertices are included in $Y$.  
        If $Y$ includes exactly one side vertex $t$, then the only monochromatic sides connected to $Y$ are the two adjacent to $t$. 
        Therefore, the connected component including $t$ must contain a cycle by Lemma~\ref{lem:useful}, contradicting the fact that $Y$ is a tree.  
        If $Y$ includes no side vertex of $T$, then its structure must be one of the two cases shown in Figure~\ref{fig:tree-B-1p}.  
        By Lemma~\ref{lem:alternate}, we obtain the following equations:
        \begin{itemize}
            \item \textbf{Equations A:} $\| x_1 \| + \| y_2 \| = \sigma$, \quad $\| y_1 \| + \| z_2 \| = \sigma$, \quad $\| z_1 \| + \| x_2 \| = \sigma$.
            \item \textbf{Equations B:} $\| x_1 \| = \| y_2 \|$, \quad $\| y_1 \| = \| z_2 \|$, \quad $\| z_1 \| + \| x_2 \| = \sigma$.
        \end{itemize}
        Here, even if $Y$ contains a flat node, Lemma~\ref{lem:U-shape} implies that the side of the piece being divided has length $\sigma$.  
        From Equations A, we derive:
        $$\| x_1 \| + \| x_2 \| + \| y_1 \| + \| y_2 \| + \| z_1 \| + \| z_2 \| = 3\sigma,$$
        which contradicts Equation 1 and 2 since $\sigma \neq \tau$.  
        Similarly, Equations B also contradicts Equation 1 and 2.
        Specifically, when Equations B and Equation 1 hold, we have $\| z_1 \| + \| x_2 \| = \sigma$ and $\| z_1 \| + \| x_2 \| = \tau$, which is a contradiction. 
        When Equations B and Equation 2 hold, we have $\| z_1 \| + \| y_2 \| = \sigma$ and $\| z_1 \| + \| y_2 \| = \tau$, which is also a contradiction.
        
        \textbf{\boldmath Case B-2: The degree-3 vertex $q$ is in $G^S$:}  
        Here, $q$ is a boundary vertex, and the divided sides are denoted as $e_1$ and $e_2$. Let $f_1$ and $f_2$ be their respective endpoints.  
        Without loss of generality, we assume that $x_2$ and $e_1$, $z_1$ and $e_2$, $x_1$ and $y_2$, and $y_1$ and $z_2$ share along paths.  
        Since $q$ is in $G^S$, all surrounding angles are less than $\pi$, and each $W_-$ contains at least one non-paired node. 
        The first non-paired node encountered is either $q$ or an element of $N^{S}_{\mathit{dit}}(\VD)$.  
        
        If $Y$ includes two side nodes $t$ and $t^{\prime}$ on sides of $T$, then $W_x$ and $W_z$ contain monochromatic edges.  
        Since at least one of $v_x, v_y$, and $v_z$ is adjacent to monochromatic edges on both sides (specifically, $v_y$), $W_y$ contains no vertex on edges of $T$.  
        The possible structures are shown in Figure~\ref{fig:CaseB-2-1}.  
        Depending on whether we are considering the left or right case in the figure, by Lemma~\ref{lem:alternate}, we have one of the following:
        \begin{itemize}
            \item \textbf{Equations C:} $\| x_1 \| = \| y_2 \|$, \quad $\| y_1 \| = \| z_2 \|$, \quad $\| z_1 \| + \| x_2 \| = \sigma$.
            \item \textbf{Equations D:} $\| x_1 \| + \| y_2 \| = \tau$, \quad $\| y_1 \| + \| z_2 \| = \tau$, \quad $\| z_1 \| + \| x_2 \| = \sigma$.
        \end{itemize}
        Equations C are equivalent to Equations B, and Equations D lead to $| x_1 | + | x_2 | + | y_1 | + | y_2 | + | z_1 | + | z_2 | = 2\tau + \sigma$.
        Therefore, both contradict Equation 1 and Equation 2.
        \begin{figure}[htbp]
            \centering
            \includegraphics[width=0.5\textwidth]{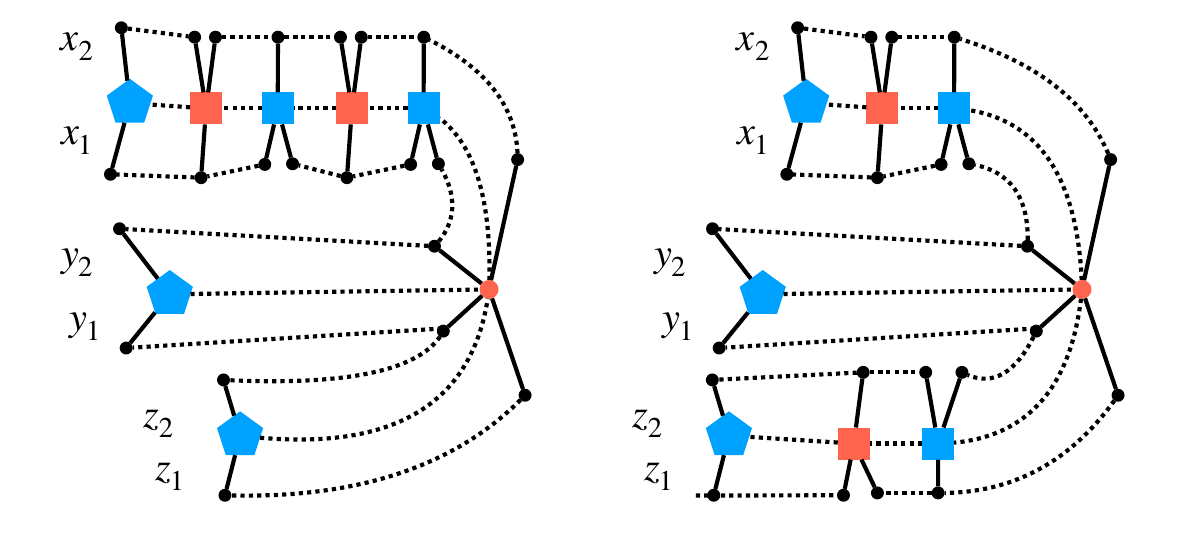}
            \caption{The possible structures for $P$ in Case B-2-1.}
            \label{fig:CaseB-2-1}
        \end{figure}   

        If $Y$ includes one side node $t$ of $T$, the only monochromatic edges connected to $Y$ are the two adjacent to $t$. These edges should be connected via paths in $\ED$, contradicting Corollary~\ref{lem:adjacent}.  

        Finally, if $Y$ includes no side node of $T$, then $Y$ contains at most one vertex from $N^{T}_{\mathit{dit}}(\VD)$, which is a flat node.  
        By Lemma~\ref{lem:U-shape}, the side length of the divided piece must be $\sigma$.
        Consequently, the possible patterns are the seven shown in Figure~\ref{fig:tree-B-2-2}.         
        This again leads to Equations A or B, contradicting Equation 1 and 2.
        \begin{figure}[htbp]
            \centering
            \includegraphics[width=\textwidth]{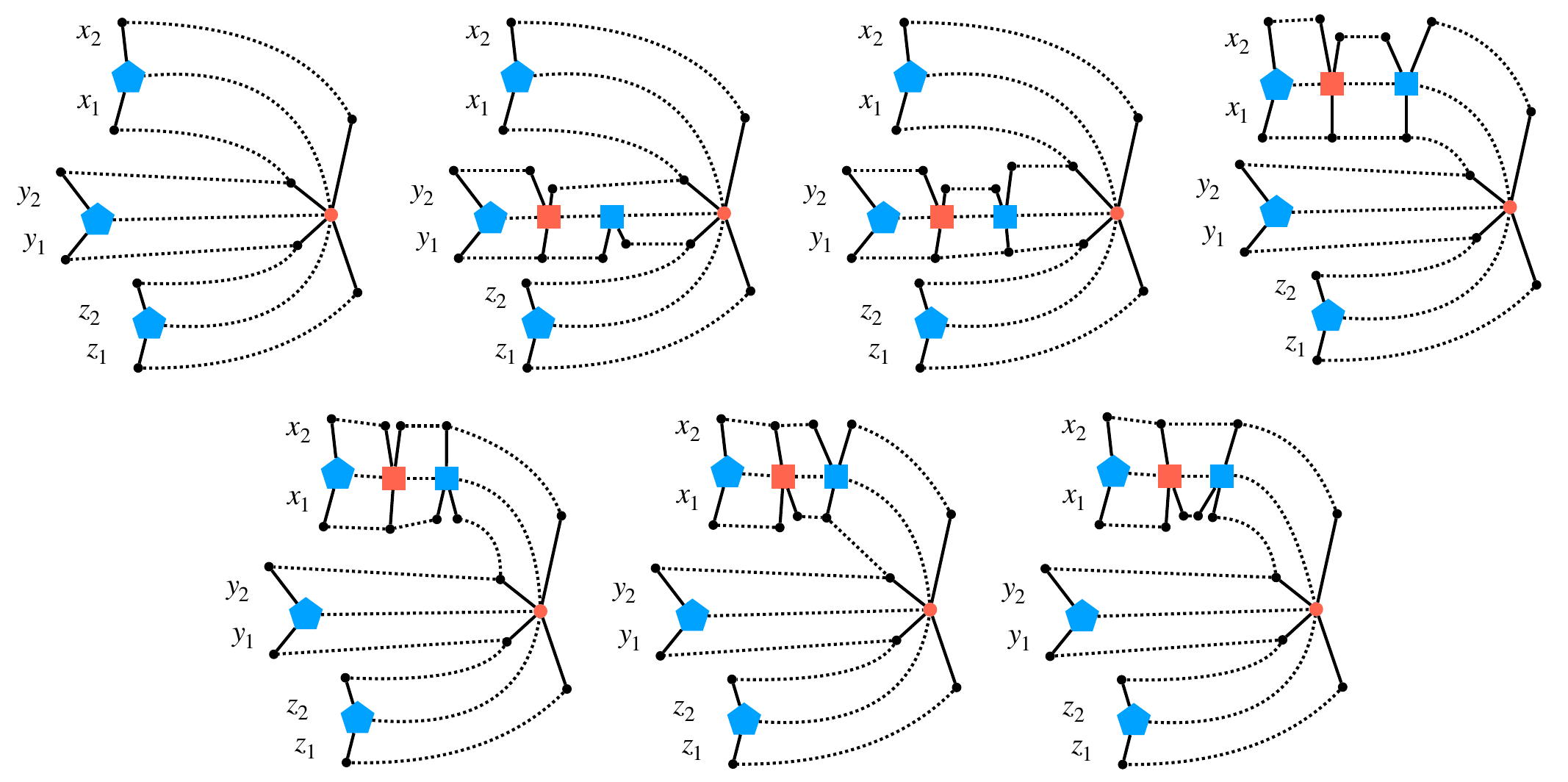}
            \caption{The possible structures for $W$ in Case B-2-3.}
            \label{fig:tree-B-2-2}
        \end{figure}
    \end{proof}
    \begin{lemma}
        Case C is infeasible.
    \end{lemma}
    \begin{proof}
        We classify the cases as follows:
        \begin{itemize}
            \item \textbf{Case C-1:} $\{G^T_{\mathfrak{D}2}, G^T_{\mathfrak{F}1}\} \times \{G^S_{\mathfrak{D}5}\}, \{G^T_{\mathfrak{D}3}, G^T_{\mathfrak{E}1}\} \times \{G^S_{\mathfrak{D}10}, G^S_{\mathfrak{F}10}\}$
            \item \textbf{Case C-2:} $\{G^T_{\mathfrak{D}1}\} \times \{G^S_{\mathfrak{A}4}, G^S_{\mathfrak{B}1}\}, \{G^T_{\mathfrak{D}3}, G^T_{\mathfrak{E}1}\} \times \{G^S_{\mathfrak{A}6}, G^S_{\mathfrak{B}3}, G^S_{\mathfrak{C}1}\}$
            \item \textbf{Case C-3:} $\{G^T_{\mathfrak{D}1}\} \times \{G^S_{\mathfrak{D}3}, G^S_{\mathfrak{D}4}\}, \{G^T_{\mathfrak{D}3}, G^T_{\mathfrak{E}1}\} \times \{G^S_{\mathfrak{D}13}, G^S_{\mathfrak{D}14}, G^S_{\mathfrak{E}5}, G^S_{\mathfrak{E}6}, G^S_{\mathfrak{E}7}\}$
        \end{itemize}
        
        By Lemmas~\ref{lem:struct_BandC} and \ref{lem:VG-component}, the structure of the connected components in each case can be classified as follows:
        \begin{itemize}
            \item \textbf{Case C-1:} 
            \begin{itemize}
                \item A well-behaved path $W$ and a tree $Y$ with a degree-3 vertex.
            \end{itemize}
            \item \textbf{Case C-2:} 
            \begin{itemize}
                \item A well-behaved path $W_1$, a path $W_2$, and some component $C$, or
                \item A well-behaved path $W_1$ and some component $C$.
            \end{itemize}
            \item \textbf{Case C-3:} 
            \begin{itemize}
                \item A well-behaved path $W$, a tadpole graph $L$, which is a graph obtained by joining a cycle graph to a path graph with a bridge of degree 3, and a tree $Y$ with a degree-3 vertex, or
                \item A well-behaved path $W$ and a tree $Y$ with a degree-4 vertex.
            \end{itemize}
        \end{itemize}
    
        \textbf{\boldmath Case C-1:} $\{G^T_{\mathfrak{D}2}, G^T_{\mathfrak{F}1}\} \times \{G^S_{\mathfrak{D}5}\}, \{G^T_{\mathfrak{D}3}, G^T_{\mathfrak{E}1}\} \times \{G^S_{\mathfrak{D}10}, G^S_{\mathfrak{F}10}\}$:
        As shown in the proof of Lemma~\ref{lem:struct_BandC}, the node in $\ED$ that shares a path with the side of $S$ where both adjacent corners have degree-2 in $\VD$ must be either a trisected-mid node or a flat node. 
        In $G^S_{\mathfrak{F}10}$, the trisected-mid node lies on the side of $S$, however its length is strictly shorter than the side length of $S$, leading to a contradiction. 
        Thus, we only consider $G^S_{\mathfrak{D}5}$ and $G^S_{\mathfrak{D}10}$.
        Since the endpoints $w$ and $w^{\prime}$ of $W$ are corner vertices of $S$, if $W$ contains a divide-in-two vertex $t$ on the side of $S$, only the two next sides of $t$ become monochromatic (Figure~\ref{fig:case-C-1}).
        This contradicts Lemma~\ref{lem:useful}.
        Thus, we only consider the case where $W$ does not contain a divide-in-two vertex on the side of $S$.
        \begin{figure}[htbp]
            \centering
            \includegraphics[width=0.25\textwidth]{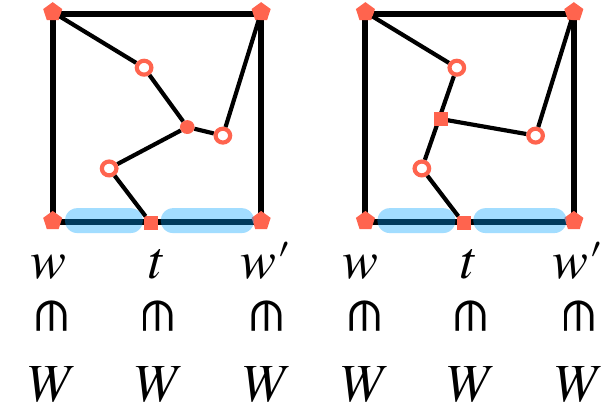}
            \caption{Case where $W$ contains a divide-in-two vertex $t$ on the side of $S$ (left: $G^S_{\mathfrak{D}5}$, right: $G^S_{\mathfrak{D}10}$).}
            \label{fig:case-C-1}
        \end{figure}

        \textbf{\boldmath Case C-1-1:} $\{G^T_{\mathfrak{D}2}, G^T_{\mathfrak{F}1}\} \times \{G^S_{\mathfrak{D}5}\}$: 
        Since $W$ does not contain a divide-in-two vertex on the side of $S$, there is exactly one divide-in-two vertex on the side of $T$ contained in $W$.
        Thus, the structure of $W$ must be one of the two cases shown in Figure~\ref{fig:case-C-1-1}. 
        In these cases, one of the along sequences consists of a single path, while the other consists of two paths.
        According to Lemma~\ref{lem:alternate}, the sum of the lengths of the paths in the along sequence consisting of two paths should equal $\tau$.
        However, the length of the counterclockwise-next side of $w$ and clockwise-next side of $w^{\prime}$ are both $\sigma$, and the sum of the lengths of the clockwise next side of $w$ and the counterclockwise next side of $w^{\prime}$ is also $\sigma$.
        As a result, it is impossible for either sum to be $\tau$, leading to a contradiction.

        \begin{figure}[htbp]
            \centering
            \includegraphics[width=\textwidth]{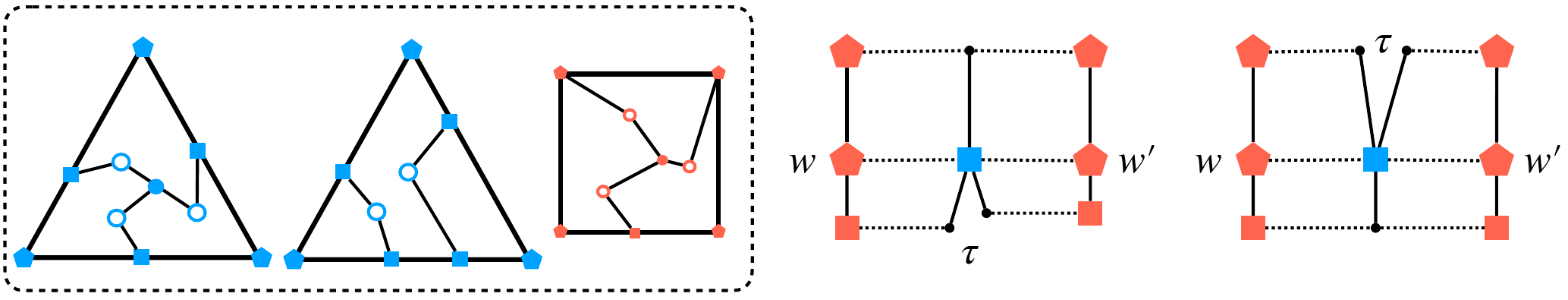}
            \caption{Possible structures in the case where $\{G^T_{\mathfrak{D}2}, G^T_{\mathfrak{F}1}\} \times \{G^S_{\mathfrak{D}5}\}$.}
            \label{fig:case-C-1-1}
        \end{figure}

        \textbf{\boldmath Case C-1-2:} $\{G^T_{\mathfrak{D}3}, G^T_{\mathfrak{E}1}\} \times \{G^S_{\mathfrak{D}10}\}$: 
        Since $W$ does not contain a divide-in-two vertex on the side of $S$, $W$ either contains only one divide-in-two vertex on the side of $T$, or it contains one flat vertex on $S$ and two divide-in-two vertices on the side of $T$.
        In the former case, a contradiction arises in the same way as in Case C-1-1.
        In the latter case, the remaining boundary vertex $t$ on $T$ that is not contained in $W$ is contained in $Y$.
        (There is no divide-in-two vertex on $S$ that is not contained in $W$, so $t$ cannot be part of a cycle.)
        Since $t$ is contained in $Y$, only its next side becomes monochromatic, which contradicts Lemma~\ref{lem:useful} (Figure~\ref{fig:case-C-1-2}).

        \begin{figure}[htbp]
            \centering
            \includegraphics[width=0.6\textwidth]{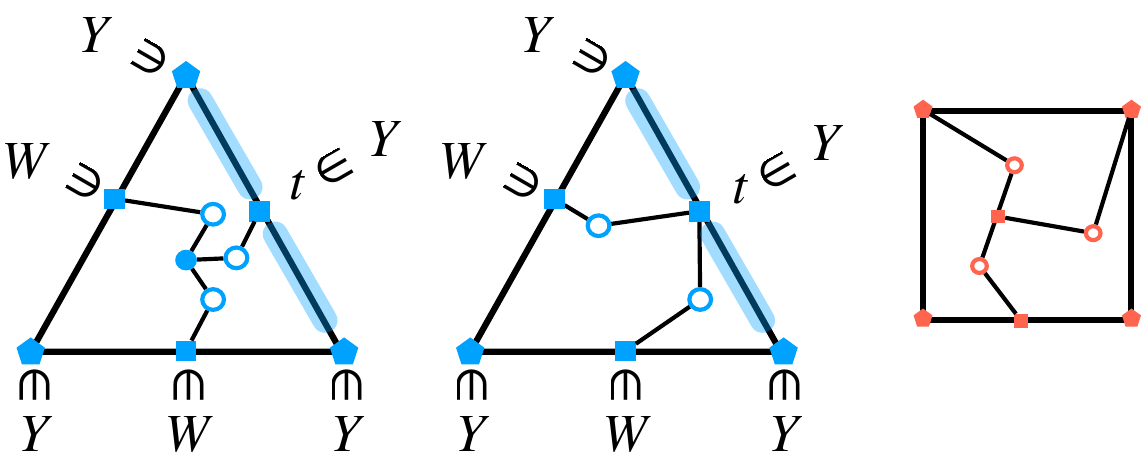}
            \caption{Possible structures in the case where $\{G^T_{\mathfrak{D}3}, G^T_{\mathfrak{E}1}\} \times \{G^S_{\mathfrak{D}10}\}$.}
            \label{fig:case-C-1-2}
        \end{figure}
    
        \textbf{\boldmath Case C-2:} $\{G^T_{\mathfrak{D}1}\} \times \{G^S_{\mathfrak{A}4}, G^S_{\mathfrak{B}1}\}, \{G^T_{\mathfrak{D}3}, G^T_{\mathfrak{E}1}\} \times \{G^S_{\mathfrak{A}6}, G^S_{\mathfrak{B}3}, G^S_{\mathfrak{C}1}\}$:
        If $\VD$ has two separate paths $W_1$ and $W_2$, a contradiction arises by Lemma~\ref{lem:U-shape}.
        Therefore, we consider the case where there is only one path $W$ contained in $\VD$.
        By Lemma~\ref{lem:U-shape}, it suffices to consider the case where the endpoints $w$ and $w^{\prime}$ of $W$ belong to different U-shaped boundaries.
        There are two possibilities: either $w$ and $w^{\prime}$ are diagonal corners of $S$, or $w$ and $w^{\prime}$ are the endpoint corners of a single side of $S$.

        \textbf{\boldmath Case C-2-1: the case where $w$ and $w^{\prime}$ are diagonal corners of $S$}: 
        If $W$ contains a divide-in-two node on the side of $S$, it contradicts Lemma~\ref{lem:useful} (see the left of Figure~\ref{fig:case-C-2-1}). 
        Conversely, if $W$ does not contain a divide-in-two node on the side of $S$, the structure of $W$ must be one of the two cases shown in the right of Figure~\ref{fig:case-C-2-1}. 
        In both cases, one of the along sequences of $W$ forms a single path in $\ED$, however the next sides of $w$ and $w^{\prime}$ in different directions have different lengths, leading to a contradiction to Observation~\ref{obs:EG-length}.

        \begin{figure}[htbp]
            \centering
            \includegraphics[width=1\textwidth]{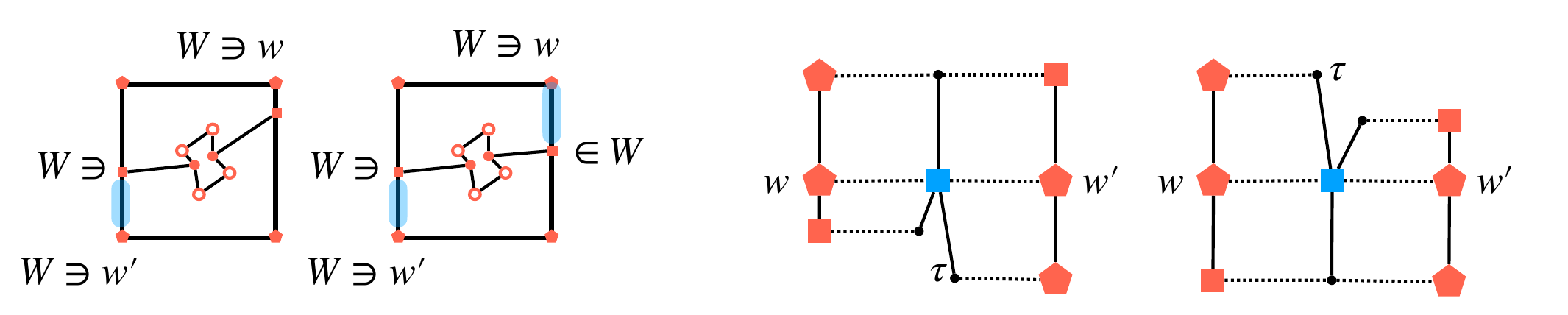}
            \caption{Possible structures in the case where $w$ and $w^{\prime}$ are diagonal corners of $S$.}
            \label{fig:case-C-2-1}
        \end{figure}

        \textbf{\boldmath Case C-2-2: the case where $w$ and $w^{\prime}$ are the endpoint corners of a single side of $S$}: 
        If $W$ contains two divide-in-two nodes on $S$, one of them divides the side whose endpoints are the corner vertices $w$ and $w^{\prime}$. The boundary edges corresponding to the divided side are both monochromatic, which contradicts Lemma~\ref{lem:useful}.
        For the other cases, where $W$ contains at most one divide-in-two node on $S$, a contradiction arises in the same way as in Case C-1-2.
        
        \textbf{\boldmath Case C-3:} $\{G^T_{\mathfrak{D}1}\} \times \{G^S_{\mathfrak{D}3}, G^S_{\mathfrak{D}4}\}, \{G^T_{\mathfrak{D}3}, G^T_{\mathfrak{E}1}\} \times \{G^S_{\mathfrak{D}13}, G^S_{\mathfrak{D}14}, G^S_{\mathfrak{E}5}, G^S_{\mathfrak{E}6}, G^S_{\mathfrak{E}7}\}$: 
        If $\VD$ contains a tadpole graph $L$, then the degree-2 corner vertex of $S$ (along with one of the degree-1 corner vertices) is included in $L$, and the tree $Y$ is well-behaved. 
        In this case, a contradiction arises using the same argument as in Case B. 
        Therefore, we focus on the scenario where $\VD$ consists of a well-behaved path $W$ and a tree $Y$ with a degree-4 vertex. 
        In each case, there exist two sides that are not divided by a side vertex, which we denote as $e$ and $e^{\prime}$.
                
        \textbf{\boldmath Case C-3-1: the case where $e$ and $e^{\prime}$ are connected by a path in $\ED$}: 
        In the cases of $\{G^T_{\mathfrak{D}1}\} \times \{G^S_{\mathfrak{D}4}\}, \{G^T_{\mathfrak{D}3}, G^T_{\mathfrak{E}1}\} \times \{G^S_{\mathfrak{D}14}, G^S_{\mathfrak{E}5}\}$, Lemma~\ref{lem:side-component} implies that a degree-2 corner vertex is included in a cycle, contradicting the fact that $\VD$ does not contain a tadpole graph.
        In other cases, such as $\{G^T_{\mathfrak{D}1}\} \times \{G^S_{\mathfrak{D}3}\}, \{G^T_{\mathfrak{D}3}, G^T_{\mathfrak{E}1}\} \times \{G^S_{\mathfrak{D}13}, G^S_{\mathfrak{E}6}, G^S_{\mathfrak{E}7}\}$, computing the along sequence of $W$ on the side that is not $e$, regardless of the number of divide-in-two vertices contained in $W$, leads to a contradiction with Lemma~\ref{lem:alternate}.

        \textbf{\boldmath Case C-3-2: the case where $e$ and $e^{\prime}$ are not connected by a path in $\ED$}: 
        Since $e$ and $e^{\prime}$ both have a length of $\sigma$, the other endpoints of the $\ED$ paths originating from each of them must be different edges on the boundary of $T$.
        Therefore, $W$ contains at least two divide-in-two vertices on the sides of $T$ and at least one divide-in-two vertex $t$ on the side of $S$.
        Let the endpoints of $W$ be $w$ and $w^{\prime}$. 
        If $t$ shares a side on the boundary of $S$ with at least one of $w$ or $w^{\prime}$, then the monochromatic edges whose endpoints are contained in $W$ must either exist in an odd number, exist in a pair with different orientations, or exist in a pair that shares an endpoint. 
        In any case, this contradicts Lemma~\ref{lem:useful}.
        Therefore, it is sufficient to consider the case where $t$ shares no side on the boundary of $S$ with $w$ and $w^{\prime}$.
        This can occur in $\{G^T_{\mathfrak{D}1}\} \times \{G^S_{\mathfrak{D}3}\}$ and $\{G^T_{\mathfrak{D}3}, G^T_{\mathfrak{E}1}\} \times \{G^S_{\mathfrak{D}13}\}$, where the endpoints of $W$ are on a single side of $S$, and $t$ is located on its opposite side (Figure~\ref{fig:case-C-3-2}).
        However, in this case, focusing on the remaining divide-in-two vertices on the side of $T$ leads to a contradiction for the same reason as in Case C-1-2.
    \end{proof}

        \begin{figure}[htbp]
            \centering
            \includegraphics[width=1\textwidth]{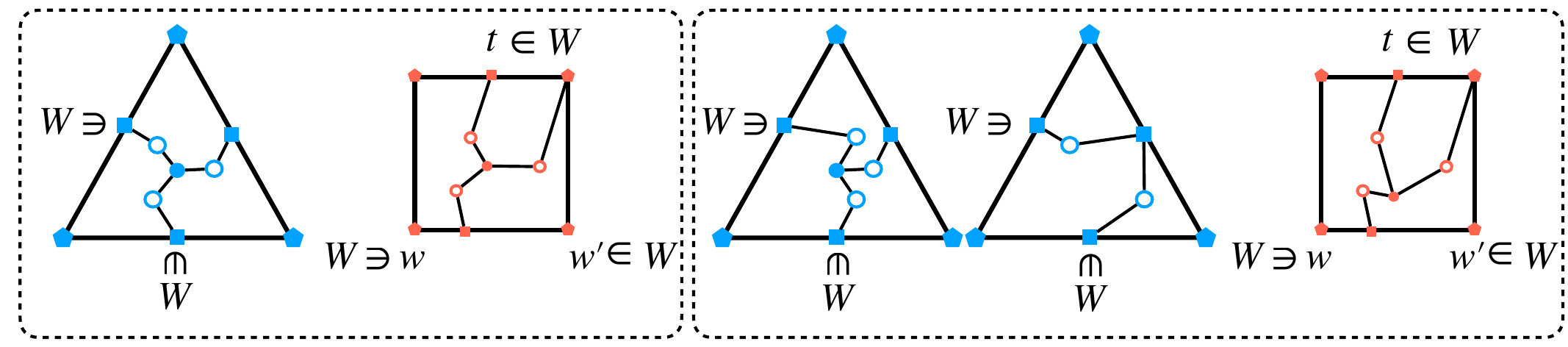}
            \caption{The case where $t$ shares no side on the boundary of $S$ with $w$ and $w^{\prime}$.
Left: $\{G^T_{\mathfrak{D}1}\} \times \{G^S_{\mathfrak{D}3}\}$; Right: $\{G^T_{\mathfrak{D}3}, G^T_{\mathfrak{E}1}\} \times \{G^S_{\mathfrak{D}13}\}$.}
            \label{fig:case-C-3-2}
        \end{figure}

\section{Conclusion}
    In this paper, we demonstrated that a three-piece dissection between a square and an equilateral triangle does not exist. 
    The concept of matching diagrams, which was central to our proof, seems to be broadly applicable as a method for analyzing dissections. 

    We believe that our approach can be extended to handle the case where pieces can be flipped over. The challenge is to do so without vastly increasing the number of cases that need to be considered.
    We highlight this and other open problems:
    \begin{itemize}
        \item Is a three-piece dissection still impossible if flipping is allowed?
        \item Is a three-piece dissection still impossible if we allow nonpolygonal (curved) pieces?
        \item Are there only a finite number of rectangles that can be dissected into three pieces from a triangle? If so, how can they be enumerated?
        \item Are there any other four-piece dissections between an equilateral triangle and a square, aside from Figure~\ref{fig:dudeney}?
        \item Are there any pairs of regular $n$-gons and $m$-gons that can be dissected into three pieces, where $n \neq m$?
    \end{itemize}
\section*{Acknowledgement}
    The authors would like to thank the members of the Tokyo-based math club “sugakuday,” a group that describes itself as a gathering of people who like math—or don’t.
    During a visit by Tonan Kamata, a discussion with the group—particularly with Motcho Ajisaka and Fueruma (\url{@motcho_tw}, \url{@fueruma_suugaku} on $\mathbb{X}$)—led to the discovery of an oversight in the enumeration of square dissections in the previous version of this paper.
    We express our deep respect for their open, playful, and curiosity-driven community.
    
    \bibliography{main}
\end{document}